\documentclass{amsart}

\usepackage[T1]{fontenc}
\usepackage{latexsym,amsfonts}
\usepackage{tikz}
\usetikzlibrary{calc,shapes,arrows,automata}
\usepackage{amsmath, amssymb}
\usepackage{algorithm}
\usepackage{algorithmic}
\usepackage{url}
\usepackage{subcaption} 
\usepackage{caption}
\usepackage{centernot}
\usepackage{epsfig}
\usepackage{thmtools}
\usepackage{hyperref}
\usepackage{multirow}
\usepackage{textcomp}
\usepackage{supertabular}
\usepackage{preamble}

\begin{document}

\begin{abstract}
    \label{sec:abstract}
    Barrier certificates, a form of \emph{state invariants}, provide an automated approach to the verification of the safety of dynamical systems. 
    Similarly to barrier certificates, recent works explore the notion of closure certificates, a form of \emph{transition invariants}, to verify dynamical systems against $\omega$-regular properties including safety.
    A \emph{closure certificate}, defined over state pairs of a dynamical system, is a real-valued function whose zero superlevel set characterizes an inductive transition invariant of the system.   
    The search for such a certificate can be effectively automated by assuming it to be within a specific template class, \textit{e.g.} a polynomial of a fixed degree, and then using optimization techniques such as sum-of-squares (SOS) programming to find it. 
    Unfortunately, one may not be able to find such a certificate for a fixed template.
    In such a case, one must change the template, \textit{e.g.} increase the degree of the polynomial.
    In this paper, we consider a notion of multiple closure certificates dubbed interpolation-inspired closure certificates.
    An interpolation-inspired closure certificate consists of a set of functions which jointly over-approximate a transition invariant by first considering one-step transitions, then two, and so on until a transition invariant is obtained.
    The advantage of interpolation-inspired closure certificates is that they allow us to prove properties even when a single function for a fixed template cannot be found using standard approaches. 
    We present SOS programming and a scenario program to find these sets of functions and demonstrate the effectiveness of our proposed method to verify persistence and general $\omega$-regular specifications in some case studies.
\end{abstract}

\title[Interpolation-Inspired Closure Certificates]{Interpolation-Inspired Closure Certificates}

\thanks{
This work was supported by NSF under grants CNS-2111688 and CNS-2145184.}

\author[mohammed adib oumer]{Mohammed Adib Oumer} 
\author[vishnu murali]{Vishnu Murali} 
\author[majid zamani]{Majid Zamani}

\address{Department of Computer Science at the University of Colorado, Boulder, CO, USA.}
\email{\{mohammed.oumer,~vishnu.murali,~majid.zamani\}@colorado.edu}
\urladdr{https://www.hyconsys.com/members/moumer/}
\urladdr{https://www.hyconsys.com/members/vmurali/}
\urladdr{https://www.hyconsys.com/members/mzamani/}

\maketitle

\section{Introduction} 
\label{sec:intro}

With the proliferation of cyber-physical systems in mission critical areas, it is crucial to formally verify that they satisfy properties of interest. 
One rigorous paradigm to frame these properties of interest is in terms of logic such as linear temporal logic (LTL) \cite{pnueli1977temporal} or $\omega$-regular automata \cite{vardi1994reasoning}.
Two useful approaches to verify systems against these specifications are through the use of barrier certificates \cite{prajna2004safety}, which act as inductive state invariants, or closure certificates \cite{murali2024closure}, which act as inductive transition invariants.
Closure certificates introduced in \cite{murali2024closure} aim to capture transitive closure of transition relations with the goal of automating the verification of LTL and $\omega$-regular specifications.
The search for such closure certificates can be automated using optimization approaches \cite{parrilo2003semidefinite,prajna2002sostools}.  
These approaches typically rely on first fixing a template and then making use of the above approaches to search for an appropriate function in this template. 
When one is unable to find such a certificate for a given template, one is forced to consider a different template and try again. 
Inspired by the success of interpolation \cite{mcmillan2003interpolation} in the verification of software and hardware systems, and their adoption of barrier certificates in interpolation-inspired barrier certificates \cite{oumer2024ibc}, we present a notion of \emph{interpolation-inspired closure certificates}.
In this paper, we show that one may find such interpolation-inspired certificates to verify LTL properties such as safety and persistence, even when we fail to find a single function to do so. 
The search for such interpolation-inspired certificates can be carried out effectively through existing approaches.

Verification of systems against safety, a prominent LTL property, can be automated using barrier certificates~\cite{prajna2004safety}.
A barrier certificate is a real valued function that is nonpositive over the initial states, positive over the unsafe states, and nonincreasing with transitions. 
Thus, a barrier certificate guarantees safety as its zero level set separates the reachable and unsafe states with the zero-sublevel set over-approximating the set of reachable states.
The authors of \cite{wongpiromsarn2015automata} extended the barrier certificate-based approach to refute violations of linear temporal logic (LTL) specifications expressed via $\omega$-automata, although such an approach is conservative. 
An approach to remedy this proposed in \cite{murali2023co}, was to show that a certain set is visited only a fixed number of times.
Unfortunately, such an approach relied on a user-selected hyperparameter to limit the number of visits to this set.
Another approach to address this is to consider a well-foundedness argument\cite{alpern1987recognizing} over the transitive closure of the transition relation.

Transition invariants~\cite{podelski2004transition} act as an over-approximation of the transitive closure of the transition relation and are used to verify programs against $\omega$-regular properties. 
The authors of \cite{murali2024closure} introduced closure certificates, to serve as functional analogs of inductive transition invariants. 
Such certificates may be used to verify larger classes of $\omega$-regular properties.
Intuitively, a closure certificate is a real valued function that is nonnegative for a pair of states $(x,y)$ if $y$ may be reachable from $x$.
Ensuring that this certificate is negative for all pairs $(x_0, x_u)$, where $x_0$ is an initial state, and $x_u$ is an unsafe state, thus acts as a proof of safety.
To verify general LTL specifications, one may extend the earlier one with a ranking function argument \cite{alpern1987recognizing} to prove the well-foundedness.

In the context of dynamical systems, typical formulations for closure certificates assume the presence of a single inductive transition invariant. 
When those conditions are not met, one cannot validate or refute the property of interest for the system.
A common approach to address this challenge is to incrementally strengthen the property as a combination of other properties.
One method of adopting such an incremental strengthening approach is that of interpolation-based model checking \cite{mcmillan2003interpolation}.
This paper utilizes ideas inspired by (logical) interpolation to find multiple functions incrementally that together serve as a closure certificate to prove $\omega$-regular properties for discrete-time dynamical systems. 
This allows us to find simpler functions to ease the verification process.

\vspace{0.5em} \noindent \textbf{Related works.}
Inductive invariants and incremental inductive proofs have served as important methods in verifying the safety of finite state-transition systems, as demonstrated in \cite{zhang2004incremental, cabodi2008strengthening, bradley2011sat}. 
The goal of safety verification is to ensure that the negation of the formula representing the unsafe states is an invariant.
Unfortunately, this is often not inductive (i.e., $i)$ it holds true for the initial states, and $ii)$ it holds true for the next subsequent state if it is true for the current state). Thus, the formula is incrementally strengthened to become inductive via interpolation-based approaches~\cite{mcmillan2003interpolation, bradley2012understanding}. 
Given a property of interest, such proofs try to incrementally constrain the formula until an inductive proof is obtained. 
In the context of bounded model checking, interpolation unrolls the transition function a fixed number of times and finds intermediate formulae called interpolants until an inductive invariant formula is found. This approach was adopted in \cite{oumer2024ibc} for the verification of the safety of dynamical systems using barrier certificates.

The authors of~\cite{prajna2004safety} proposed the notion of barrier certificates as a discretization-free approach to give guarantees of safety~\cite{prajna2007convex} for dynamical and hybrid systems.
Another common framework for formulating safety is the use of invariant sets in the form of barrier functions \cite{ames2019control}. They fix the invariant set instead of searching for it and then utilize it to design a controller.
The results in~\cite{wongpiromsarn2015automata} presented an approach that uses barrier certificates to verify the LTL properties specified by $\omega$-automata. The approach has been extended for verification and synthesis for more general dynamical systems~\cite{jagtap2018temporal,jagtap2020formal}.
The authors of \cite{dimitrova2014deductive,chatterjee2024sound} consider the use of barrier certificates and ranking functions to verify LTL properties. 
The authors of~\cite{podelski2004transition} proposed a notion of transition invariants and demonstrated their use in verifying programs against $\omega$-regular properties. 
This concept was adopted in~\cite{murali2024closure} for an automated approach to verify dynamical systems against $\omega$-regular properties.

This paper builds on the results reported in~\cite{oumer2025icc}. 
The main additional contributions of this version, relative to~\cite{oumer2025icc}, are as follows: i) We include a new section on employing interpolation-inspired closure certificates for the verification of LTL specifications; ii) We introduce a section on scenario programs as a data-driven approach to compute interpolation-inspired closure certificates; iii) We replace the previous safety-focused case study section with one addressing a general LTL specification; iv) We revise our persistence case study to consider a three-dimensional system.

\section{Preliminaries}
\label{sec:prelims}
\subsection{Notation}
\label{subsec:notation}
We use $\N$ and $\R$ to denote the set of natural numbers and reals, respectively. 
For $k \in \R$, we use $\R_{\geq k}$ and $\R_{> k}$ to denote the intervals $[k, \infty)$ and $(k,\infty)$, respectively. Similarly, for any natural number $n \in \N$, we use $\N_{\geq n}$ to denote the set of natural numbers greater than or equal to $n$. 
Given two numbers $a,b \in \R$, we use $\min(a,b)$ and $\max(a,b)$ to denote their minimum and maximum, respectively.
The $n$-dimensional Euclidean space is denoted by $\R^n$.
The infinity norm ($\infty$-norm) of $x \in \R^n$ is denoted by $\|x\|$.
Symbols $\forall$ and $\exists$ denote the universal and existential quantifiers, respectively. 

Given a set $X$, we use $|X|$ to denote its cardinality. 
Given a relation $R \subseteq X \times Y$ and an element $x \in X$, we use $R(x)$ to denote the set $\{ y \mid (x,y) \in R \}$. 
For $x, y \in X$, $[x;y]$ denotes the concatenation of $x$ and $y$ into a single vector.
The set $X^{\omega}$ denotes the set of countably infinite sequences of elements in $X$.
We use the notation 
$(x_0, x_1, \ldots )\in X^{\omega}$ for $\omega$-sequences.
Let $\Inf(s)$ be the set of infinitely often occurring elements in the sequence $s  =(x_0, x_1, \ldots )$.
That is, given a sequence $s = (x_0, x_1, \ldots )$, we say that $x \in \Inf(s)$, if for all $i \in \N$, there exists $j \geq i$, such that $x_j = x$.
Given an infinite sequence $s = (x_0, x_1, \ldots )$ and two natural numbers $i,j \in \N$ where $i \leq j$, we use $s[i,j]$ to indicate the finite sequence $(x_i, x_{i+1}, \ldots, x_j)$, and $s[i, \infty)$ to indicate the infinite sequence $( x_i, x_{i+1}, \ldots )$. 
Finally, we use $s[i]$ to denote the $i^{th}$ element in the sequence $s$, \textit{i.e.}, we have $s[i] = x_i$ for any $i \in \N$.

\subsection{Discrete-time Dynamical System}
\label{subsec:system}
In this paper, we deal with discrete-time dynamical systems.
\begin{definition}
\label{def:system}
    A discrete-time dynamical system is given by the tuple $\Sys = (\Xx, \Xx_0, f)$,
    where the state set is denoted by $\Xx$, the set of initial states is $\Xx_0 \subseteq \Xx$, and $f\subseteq  \Xx \times \Xx$ is the state transition relation that describes the evolution of the states of the system. 
    That is, for $x_t$, the state of the system at time step $t\in \N$, the state of the system in the next time step is given by $x_{t+1} \in f(x_t)$, $\forall x_t \in \Xx.$
\end{definition}

If $\forall x \in \Xx$, we have $|f(x)| = 1$, then we consider the transition relation $f$ to be a deterministic state transition function.
We use $f$ for both a set-valued map when it is a relation and a transition function when it is a function.
For notational convenience, given a state $x \in \Xx$, we use $x'$ to indicate a state in $f(x)$ (\textit{i.e.,} $x' \in f(x)$). 
Furthermore, we use $x_0 \in \Xx_0$ to denote a state in the initial set of states $\Xx_0$.
A state sequence is an infinite sequence  $(x_0, x_1, \ldots) \in \Xx^{\omega}$ where $x_0 \in \Xx_0$, and $x_{i+1} \in f(x_i)\ \forall i \in \N$.
We say $x_j$ is reachable from $x_i$ if $i,j \in \N, j > i$ and $x_i, x_j \in (x_0, x_1, \ldots)$ for some state sequence $(x_0, x_1, \ldots)$.

\subsection{Formal Specifications and Verification Results}
\label{subsec:specifications}
In this paper, we study increasingly complex specifications from safety to LTL over discrete-time dynamical systems.

\subsubsection{Safety}
We say that a system $\Sys$ starting from some set of initial states $\Xx_0$ is safe with respect to a set of unsafe states $\Xx_u \subseteq \Xx$ if for any state sequence $(x_0, x_1, \ldots)$, $x_i \notin \Xx_u\ \forall i \in \N_{\geq 1}$.
Remark that we assume that the initial and unsafe sets are disjoint, i.e. $\Xx_0 \cap \Xx_u = \emptyset$.
For simplicity, we use $x_u \in \Xx_u$ to denote a state in the unsafe set of states $\Xx_u$.
\begin{definition}[BC]
\label{def:bc}
   A function $\Bb: \Xx \to \R$ is a BC for a system $\Sys$ with respect to a set of unsafe states $\Xx_u$ if $\forall x \in \Xx$, $\forall x' \in f(x)$, $\forall x_0 \in \Xx_0$, and $\forall x_u \in \Xx_u$: 
    \begin{align}
        \label{eq:bc_1}
        & \Bb(x_0) \leq 0, \\
        \label{eq:bc_2}
        & \Bb(x_u) > 0, \text{ and } \\
        \label{eq:bc_3}
        & \big( \Bb(x) \leq 0 \big) {\implies} \big( \Bb(x') \leq 0 \big).
    \end{align}
\end{definition}
\begin{theorem}[BCs imply safety \cite{prajna2004safety}]
    For a system $\Sys$ with unsafe states $\Xx_u$, the existence of a barrier certificate $\Bb$ satisfying conditions~\eqref{eq:bc_1}-\eqref{eq:bc_3} implies its safety.
\end{theorem}

Building on the notion of interpolation \cite{mcmillan2003interpolation}, the authors of \cite{oumer2024ibc} proposed a notion of \emph{interpolation-inspired barrier certificates (IBCs)}.
Such certificates consist of a set of functions that together ensure safety.
\begin{definition}[IBC]
\label{def:ibc}
A set of functions $\Bb_i: \Xx \rightarrow \R$,  $\forall  i \in \{0,\ldots, k \}$, is an IBC for a system $\Sys$ with respect to a set of unsafe states $\Xx_u$ if $\forall x \in \Xx$, $\forall x' \in f(x)$, $\forall x_0 \in \Xx_0$, and $\forall x_u \in \Xx_u$:
\begin{align}
    \label{eq:ibc_1}
    &\Bb_0(x_0) \leq 0, && \\
    \label{eq:ibc_2}
    &(\Bb_i(x) \leq 0) {\implies} (\Bb_{i+1}(x') \leq 0),
    \forall i \in \{0,{\ldots}, (k{-}1) \},\\
    \label{eq:ibc_3}
    &(\Bb_k(x) \leq 0) {\implies} (\Bb_{k}(x') \leq 0),    \text{ and }\\
    \label{eq:ibc_4}
    &\Bb_i(x_u) > 0, 
    \quad \forall  i \in \{0,\ldots,k \}.
\end{align}
\end{definition}
\begin{theorem}[IBCs imply safety \cite{oumer2024ibc}]
The existence of an IBC $\Bb_i: \Xx \rightarrow \R$,$\forall i \in \{0,\ldots,k \}$, satisfying conditions~\eqref{eq:ibc_1}-\eqref{eq:ibc_4} for a system $\Sys$ with unsafe states $\Xx_u$ implies its safety.
\label{thm:ibc}
\end{theorem}

Note that when we consider IBCs with $k = 0$, one recovers the definition of BCs in Definition \ref{def:bc}.
Now, we discuss existing formulations of closure certificates.
The formulation of \emph{closure certificates (CCs)}~\cite{murali2024closure} for safety is given below.
\begin{definition}[CC for Safety]
\label{def:cc_safe}
A  function $\Tt: \Xx \times \Xx \to \R$ is a CC for a system  $\Sys$ with respect to a set of unsafe states $\Xx_{u}$ if $\exists \eta \in \R_{ > 0}$ such that $\forall x, y \in \Xx$, $\forall x' \in f(x)$, $\forall x_0 \in \Xx_0$, and $\forall x_u \in \Xx_u$:
\begin{align}
    \label{eq:cc_safe1} 
    & \Tt(x, x') \geq 0, \\
    \label{eq:cc_safe2}
    & \big( \Tt(x', y) \geq 0 \big) \implies \big( \Tt(x, y) \geq 0 \big), \text{ and } \\
    \label{eq:cc_safe3}
    & \Tt(x_0, x_u) \leq - \eta. 
\end{align}
\end{definition}

\begin{theorem}[CCs imply safety \cite{murali2024closure}]
    \label{thm:cc_safe}
   The existence of a function $\Tt: \Xx \times \Xx \to \R$ satisfying conditions~\eqref{eq:cc_safe1}-\eqref{eq:cc_safe3} for a system $\Sys$ with unsafe states $\Xx_{u}$ implies its safety.
\end{theorem}

\subsubsection{Persistence}
We say that a system visits a region $\Xx_{VF} \subseteq \Xx$ only finitely often if for any state sequence $( x_0, x_1, \ldots)$, $\exists i \in \N$ such that $\forall j \geq i$, $j \in \N$, we have $x_j \notin \Xx_{VF}$. 
To ensure persistence, one may use a Lyapunov-like argument \cite{podelski2006model} combined with an invariant argument via CCs.

\begin{definition}[CC for Persistence]
\label{def:cc_persistence}
    A function $\Tt: \Xx \times \Xx \to \R$ is a CC for $\Sys$ with a set of states $\Xx_{VF} \subseteq \Xx$ that must be visited only finitely often if $\exists \eta \in \R_{>0}$ such that $\forall x,y \in \Xx$, $\forall y', y'' \in \Xx_{VF}$, $\forall x' \in f(x)$, and $\forall x_0 \in \Xx_0$: 
    \begin{align}
        \label{eq:cc_pers1}
        &\Tt(x, x') \geq 0, \\
        \label{eq:cc_pers2}
        & \big( \Tt(x', y) \geq 0 \big) \implies \big( \Tt(x, y) \geq 0 \big), \text{ and }\\
        \label{eq:cc_pers3}
        &\big( \Tt(x_0, y') \geq 0 \big) \wedge \big( \Tt(y',y'') \geq 0 \big) \implies \nonumber \\
        &\qquad \big( \Tt(x_0,y'') \leq \Tt(x_0,y') - \eta \big). 
    \end{align}
\end{definition}
\begin{theorem}[CCs imply Persistence \cite{murali2024closure}]
    \label{thm:cc_persistence}
    The existence of a function $\Tt: \Xx \times \Xx \to \R$ satisfying conditions~\eqref{eq:cc_pers1}-\eqref{eq:cc_pers3} for a system $\Sys$ implies that the state sequences of the system visit the set $\Xx_{VF}$ finitely often.
\end{theorem}

\subsubsection{Linear Temporal Logic (LTL)}
The formulae in LTL~\cite{pnueli1977temporal}  are defined with respect to a set of finite atomic propositions $AP$ that are relevant to our system. 
Let $\Sigma = 2^{AP}$ denote the power set of atomic propositions.
A trace $w = ( w_0, w_1, \ldots, ) \in \Sigma^{\omega}$ is an infinite sequence of sets of atomic propositions.
The syntax of LTL can be given via the following grammar: 
\[
\phi := \top \;|\; a \;|\; \neg \phi \;|\; \mathsf{X} \phi  \;|\; \phi \mathsf{U} \phi ,
\]
where $\top$ indicates $\mathsf{true}$, $a \in AP$ denotes an atomic proposition, symbols $\wedge$, $\neg$ denote the logical AND and NOT operators, respectively.
The temporal operators next, and until, are denoted by $\mathsf{X}$, and $\mathsf{U}$, respectively.
The above operators are sufficient to derive the logical OR ($\vee$) and the implication ($\implies$) operators, and the temporal operators release ($\mathsf{R}$), eventually ($\mathsf{F}$) and always ($\mathsf{G}$).

We inductively define the semantics of an LTL formula with respect to trace $w$ as follows:
\begin{align}
&w \models a && \text{ if } a \in w[0] \\
& w \models \phi_1 \wedge \phi_2 && \text{ if } w \models \phi_1 \text{ and } w \models \phi_2 \\
& w \models \neg \phi && \text{ if } w \not\models \phi \\
& w \models \mathsf{X} \phi & & \text{ if } w[1, \infty) \models \phi\\
& w \models \phi_1 \mathsf{U} \phi_2 && \text{ if }\exists i \in \N \text{ such that } w[0,i] \models \phi_1 \nonumber \\ & && \text{ and } w[i+1, \infty) \models \phi_2
\end{align}

To reason about whether a system satisfies a property specified in LTL, we associate a labeling function $\Ll: \Xx \to \Sigma$ which maps each state of the system to a letter in the finite alphabet $\Sigma$. 
This naturally generalizes to mapping a state sequence of the system $( x_0, x_1, \ldots) \in \Xx^{\omega}$ to a trace $w = ( \Ll(x_0), \Ll(x_1), \ldots) \in \Sigma^{\omega}$. 
Let $TR(\Sys, \Ll)$ denote the set of all traces of $\Sys$ under the labeling map $\Ll$. 
Then the system $\Sys$ satisfies an LTL property $\phi$ under labeling map $\Ll$ if for all $w \in TR(\Sys, \Ll)$, we have $w \models \phi$.
We denote this as $\Sys \models_{\Ll} \phi$ and infer the labeling map from context.
Safety and persistence can be formulated as LTL formulae.

\vspace{0.3em}\noindent \textbf{Nondeterminstic B\"uchi Automata.}
A nondeterminstic B\"uchi automaton (NBA) $\Aa$ is a tuple $(\Sigma,\Qq, \Qq_0, \delta, \Qq_F)$, where:
$\Sigma$ is the alphabet,  $\Qq$ is a finite set of states, $\Qq_0 \subseteq \Qq$ is an initial set of states,  $\delta \subseteq \Qq \times \Sigma \times \Qq$ is the transition relation, and ${\Qq_F} \subseteq \Qq$ is the set of accepting states.
A run of the automaton $\Aa = (\Sigma,\Qq, \Qq_0, \delta, \Qq_F)$ over a trace $w = ( \sigma_0, \sigma_1, \sigma_2 \ldots, ) \in \Sigma^{\omega}$, is an infinite sequence of states  characterized as $\rho = ( q_0,q_1, q_2, \ldots, ) \in \Qq^{\omega}$ with $q_0 \in \Qq_0$ and $q_{i+1} \in \delta(q_i, \sigma_i)$.
An NBA $\Aa$ is said to accept a trace $w$ if there exists a run $\rho$ over $w$ where $\Inf(\rho) \cap {\Qq_F} \neq \emptyset$.

It is well known that given an LTL formula $\phi$ over a set of atomic propositions $AP$, one can construct an NBA $\Aa$ such that $w \in \Ll(\Aa)$ if and only if $w \models \phi$~\cite{vardi2005automata}.
An automata-theoretic technique to determine whether $\Sys \models_{\Ll} \phi$ is to first find the NBA $\Aa$ that represents $\neg \phi$, and then ensure that $\Sys \not\models_{\Ll} \neg \phi$ by showing that no trace of the system is accepted by the NBA $\Aa$.
While converting an LTL formula to an NBA is exponential in the size of the formula, negating an LTL formula has a complexity that is linear in its size.

To verify whether a given system satisfies a desired LTL property, a closure certificate is used on the product $\Sys \otimes \Aa$.

\begin{definition}[CC for LTL]
\label{def:cc_ltl}
   Consider a system $\Sys = (\Xx, \Xx_0, f)$ and NBA $\Aa= (\Sigma, \Qq, \Qq_0, \delta, \Qq_F)$ representing the complement of an LTL formula $\phi$. 
   A function $\Tt: \Xx \times \Qq \times  \Xx \times \Qq \to \R$ is a CC for $\Sys \otimes \Aa$ if $\exists \eta \in \R_{> 0}$ such that $\forall x, y, y' \in \Xx$, $\forall x' \in f(x)$, $\forall x_0 \in \Xx_0$, $\forall n, p \in \Qq$, $\forall n' \in \delta(n, \Ll(x))$, $\forall q_0 \in \Qq_0$ and $\forall r, r' \in \Qq_F$: 
    \begin{align}
        \label{eq:cc_ltl1}
        & \Tt\big( (x, n), (x', n') \big) \geq 0, \\
        \label{eq:cc_ltl2} 
        & \Big( \Tt \big( (x', n'), (y, p) \big) \geq 0 \Big) {\implies} \Big( \Tt \big((x, n), (y, p) \big) \geq 0 \Big), \\
        \label{eq:cc_ltl3}
        & \Big( \Tt\big((x_0,q_0),(y,r) \big) \geq 0 \Big) \wedge \Big( \Tt( (y,r), (y',r') ) \geq 0 \Big) \implies \nonumber \\
        & \qquad \Big(\Tt \big( (x_0,q_0),(y',r') \big) \leq \Tt \big( (x_0,q_0),(y,r) \big) - \eta \Big).  
    \end{align}
\end{definition}
\begin{theorem}[CCs verify LTL~\cite{murali2024closure}]
\label{thm:cc_ltl} 
   Consider a system $\Sys$ and an LTL formula $\phi$. Let NBA $\Aa$ represent the complement of the specification, \textit{i.e}, $\neg \phi$. The existence of a CC that satisfies the conditions~\eqref{eq:cc_ltl1}-\eqref{eq:cc_ltl3} implies that $\Sys \models_{\Ll} \phi$. 
\end{theorem}

In the next section, we use the idea of 
interpolation to define a notion of multiple closure certificates.

\section{Interpolation-Inspired Closure Certificates}
\label{sec:icc}
Here, we introduce a notion of interpolation-inspired closure certificates (ICCs) and demonstrate their efficacy in verifying safety, persistence, and general LTL specifications.

\subsection{ICCs for Safety}
\label{subsec:icc_safety}
We first define ICCs for the safety of a system $\Sys$ as in Definition \ref{def:system}.

\begin{definition}[ICC for Safety]
\label{def:icc_safe} 
A set of functions
$\Tt_i: \Xx \times \Xx \to \R$, $\forall i \in \{0,\ldots,k \}$, is an ICC for a system  $\Sys$ with respect to a set of unsafe states $\Xx_{u}$ if $\exists \eta \in \R_{>0}$ such that $\forall x, y \in \Xx$, $\forall x' \in f(x)$, $\forall x_0 \in \Xx_0$, and $\forall x_u \in \Xx_u$:
\begin{align}
    \label{eq:icc_safe1}
    &\Tt_0(x, x') \geq 0, \\
    \label{eq:icc_safe2}
    & \big(\Tt_i(x', y) \geq 0 \big) \implies \big( \Tt_{i+1}(x, y) \geq 0 \big),  \forall i \in \{0,\ldots,(k-1) \},\\
    \label{eq:icc_safe3}
    & \big(\Tt_k(x', y) \geq 0 \big) \implies \big( \Tt_k(x, y) \geq 0 \big), \text{ and }\\
    \label{eq:icc_safe4} 
    &\Tt_i(x_0, x_u) \leq - \eta, \quad \forall i \in \{0,\ldots,k \}.
\end{align}
\end{definition}

The existence of an ICC implies the safety of a system $\Sys$.

\begin{theorem}[ICCs imply Safety]
\label{thm:icc_safe}
Consider a system $\Sys = (\Xx, \Xx_0, f)$ and a set of unsafe states $\Xx_u$.
The existence of a set of functions
$\Tt_i: \Xx \times \Xx \to \R,\forall i \in \{0,\ldots,k \}$, satisfying conditions~\eqref{eq:icc_safe1}-\eqref{eq:icc_safe4} implies its safety.
\end{theorem}
\begin{proof}
Let us assume that there exists a state sequence of the system $( x_0,\ldots, x_u , \ldots )$ that reaches an unsafe state $x_u \in \Xx_u$ from some initial state $x_0$.
From condition~\eqref{eq:icc_safe1}, we have 
$\Tt_0(x_j, x_{j+1}) \geq 0,\forall j \in \N$. From condition~\eqref{eq:icc_safe2} and induction, we have $\Tt_i(x_{j-i}, x_{j+1}) \geq 0$ for all $0\leq i \leq k, j\geq i$. 
Thus, if $x_{j+1} = x_u$ for some $j \leq k$, there is some $j = i$ such that $\Tt_i(x_0, x_u) \geq 0$. If $x_{j+1} = x_u$ for some $j>k$, we must have $\Tt_k(x_0, x_u) \geq 0$ from condition~\eqref{eq:icc_safe3} and induction. Together, we get $\Tt_i(x_0, x_u) \geq 0$ for \emph{some} $0\leq i \leq k$.
According to condition~\eqref{eq:icc_safe4}, 
$\Tt_i(x_0, x_u) \leq - \eta$ for \emph{all} $0\leq i \leq k$, where $\eta \in \R_{>0}$, which is a contradiction. 
\end{proof}

We now show that any system whose safety can be ensured using IBCs can also have its safety ensured using ICCs.
\begin{theorem}
    \label{thm:ibc2icc}
    Consider a system $\Sys = (\Xx, \Xx_0, f)$, with an unsafe set of states $\Xx_u$.
    Given an IBC $\Bb_i: \Xx \to \R$, $\forall i \in \{0,\ldots,k_b \}$, (Definition~\ref{def:ibc}), there exists an ICC $\Tt_i: \Xx \times \Xx \to \R$, $ \forall i \in \{0,\ldots,(k_b-1) \}$ (Definition \ref{def:icc_safe}).
\end{theorem}
\begin{proof}
To prove this result, we make use of the  following observation.
Suppose there exists a state sequence $(x_0, x_1, \ldots, x_k, \ldots, x_{p}, \ldots)$ where $x_k =z$, and $p = \ell + (i+1)$ for some $i \in \N$. 
Then, if $\Bb_{\ell}(z) \leq 0$ for some $0\leq \ell \leq k_b$, by induction, we must have $\Bb_{m}(x_p) \leq 0$ where $m = \min(k_b, p)$. 
Thus, the function $\Bb_{m} $ over-approximates $(i+1)$-step reachable states from $x_{k}$. 
Observe that if $\ell = k_b$, then we have $m = k_b$.
Now, we demonstrate how one may construct ICCs, by considering $(i+1)$-steps from the IBC function $\Bb_{\ell}$ to the IBC function $\Bb_m$.
We define $\Tt_i: \Xx \times \Xx \to \R$  as:
 \[
\Tt_i(x,y) = 
\begin{cases}
0 & \text{if } \Bb_\ell(x) > 0 \text{ or } \Bb_m(y) \leq 0 \\& \text{ for all } 0 \leq \ell \leq k_b,  \\
& \text{ where } m = \min(k_b, \ell + (i + 1) ),\\
-\eta & \text{otherwise},
\end{cases}
\]
 where $\eta \in \R_{ > 0}$ is a positive value.
We now show that $\Tt_i$ is an ICC.
First, suppose that $\Tt_0(x, x') < 0$ for some $x \in \Xx$.
For this to be true, we must have $\Bb_\ell(x) \leq 0$, and $\Bb_{m}(x') > 0$ for some $0 \leq \ell \leq k_b$, where $m = \min(k_b, \ell + 1 )$. 
If $\ell \leq k_b - 1$, then we have $\Bb_{\ell}(x) \leq 0$, and $\Bb_{\ell+1}(x') > 0$, which contradicts condition~\eqref{eq:ibc_2}.
If $\ell = k_b$, then  we have $\Bb_{k_b}(x) \leq 0$, and $\Bb_{k_b}(x') > 0$, which contradicts condition~\eqref{eq:ibc_3}.
So $\Tt_0$ must satisfy condition~\eqref{eq:icc_safe1}.
Second, suppose that $\Tt_i(x', y) \geq 0$ and $\Tt_{i+1}(x,y) < 0$  for some $0 \leq i \leq (k_b -2)$.
By the construction of functions $\Tt_i$, if $\Tt_{i+1}(x,y) < 0$, then we must have $\Bb_{\ell}(x) \leq 0$, and $\Bb_m(y) > 0$ for some $0 \leq \ell \leq k_b$, where $m = \min(k_b,\ell + (i+2))$. 
Furthermore, as $\Tt_i(x', y) \geq 0$, 
we have that either $\Bb_{\ell'}(x') > 0$ or $\Bb_{m'}(y) \leq 0$ for all $0 \leq \ell' \leq k_b$, where $m' = \min(k_b,\ell' + (i+1))$.
Consider $\ell' = \min(k_b,\ell+1)$. 
Then we have $m' = m$ and both $\Bb_m(y) > 0$ and $\Bb_m(y) \leq 0$ must hold.
Since this isn't possible, the other alternative is if $\Bb_{\ell}(x) \leq 0 $, and $\Bb_{\ell'}(x') > 0$.
However, this contradicts with either condition \eqref{eq:ibc_2} (if $\ell < k_b$), or condition \eqref{eq:ibc_3} (if $\ell = k_b)$.
Consequently, condition~\eqref{eq:icc_safe2} holds.
One may follow a similar strategy to show that condition~\eqref{eq:icc_safe3} also holds.
Finally, consider $\Tt_i(x_0, x_u)$. Since $\Xx_0$ and $\Xx_u$ are disjoint, and from conditions~\eqref{eq:ibc_1} and~\eqref{eq:ibc_4}, we have $\Bb_0(x_0) \leq 0$, and $\Bb_j(x_u) > 0$ for all $1\leq j \leq k_b$. By definition, $\Tt_i(x_0, x_u) = - \eta$ 
 for all $0 \leq i \leq (k_b - 1)$, and hence condition~\eqref{eq:icc_safe4} holds.
\end{proof}

\subsection{ICCs for Persistence}
\label{subsec:icc_persistence}
We now show how ICCs may be used to show that a set of states is visited only finitely often. 

\begin{definition}[ICC for Persistence]
\label{def:icc_pers}
    A set of functions $\Tt_i: \Xx \times \Xx \to \R$,  $\forall i \in \{0,\ldots,k \}$, is an ICC for $\Sys$ with a set of states $\Xx_{VF} \subseteq \Xx$ that must be visited finitely often if $\exists \eta \in \R_{ > 0}$ such that $\forall x,y \in \Xx$, $\forall y', y'' \in \Xx_{VF}$, $\forall x' \in f(x)$, and $\forall x_0 \in \Xx_0$: 
    \begin{align}
        \label{eq:icc_pers1}
        &\Tt_0(x, x') \geq 0, \\
        \label{eq:icc_pers2}
        &\big( \Tt_i(x', y) \geq 0 \big) \implies \big( \Tt_{i+1}(x, y) \geq 0 \big),  \forall i \in \{0,\ldots,(k-1) \}, \\
        \label{eq:icc_pers3}
        &\big( \Tt_k(x', y) \geq 0 \big) \implies \big( \Tt_k(x, y) \geq 0 \big), \text{ and } \\
        \label{eq:icc_pers4} 
        &\big( \Tt_k(x_0, y') \geq 0 \big) \wedge \big( \Tt_k(y',y'') \geq 0 \big) \implies \nonumber \\
        & \qquad \big( \Tt_k(x_0,y'') \leq \Tt_k(x_0,y') - \eta \big).
    \end{align}
\end{definition}
We now describe the utility of ICCs for persistence.
\begin{theorem}[ICCs imply Persistence]
    \label{thm:icc_pers}
    Consider a system $\Sys = (\Xx,\Xx_0,f)$. The existence of a set of functions $\Tt_i: \Xx \times \Xx \to \R,\forall i \in \{0,\ldots,k \}$ satisfying conditions~\eqref{eq:icc_pers1}-\eqref{eq:icc_pers4} implies that the state sequences of the system visit the set $\Xx_{VF}$ only finitely often.
\end{theorem}
\begin{proof}
Suppose that there is some state sequence $ (x_0, x_1, \ldots )$ of the system that starts from state $x_0 \in \Xx_0$ and visits $\Xx_{VF}$ infinitely often.
Let the infinite sequence $(z_0, z_1, \ldots, )$ denote the states that are visited in $\Xx_{VF}$ in that order, \textit{i.e.}, the trajectory is $(x_0, \ldots, z_0, \ldots ,z_1, \ldots)$. We take $z_i = x_{j+1}$ for some $j\geq k-1$. 
Since states in $\Xx_{VF}$ are visited infinitely often, we can start from the $(k-1)^{th}$ step onward.
From conditions~\eqref{eq:icc_pers1}-\eqref{eq:icc_pers3}, and induction, we get $\Tt_k(x_{j-k}, z_i) \geq 0$, $\Tt_k(x_0, z_i) \geq 0$,  and  $\Tt_k(z_i, z_j) \geq 0$ for all indices $j > i$, $i,j \in \N$.
Using the latter two inequalities, condition~\eqref{eq:icc_pers4}, and induction, one obtains
$
\Tt_k(x_0, z_i) \leq \Tt_k(x_0, z_0) - i \eta.
$
As this is true $\forall i \in \N$, and $\eta \in \R_{ > 0}$, $\exists \ell \in \N$ such that $\Tt_k(x_0, z_\ell) < 0$.
This is a contradiction. 
\end{proof}

Setting $k = 0$ gives us a closure certificate for persistence as in Definition \ref{def:cc_persistence} and condition \eqref{eq:icc_pers2} is no longer necessary.

\subsection{ICCs for LTL Specifications}
\label{subsec:icc_ltl}
To verify whether a system satisfies a desired LTL formula $\phi$, we  first construct an NBA $\Aa = (\Sigma,\Qq, \Qq_0, \delta, \Qq_F)$ that represents the complement of the specification, $\neg \phi$.
The NBA state set is finite and, therefore, we can denote the set $\Qq$ as the set $\{0, 1, \ldots, |\Qq| - 1 \}$.
We then construct the product $\Sys \otimes \Aa = (\Xx', \Xx_0', f')$ of the system $\Sys = (\Xx, \Xx_0, f)$ with the NBA $\Aa$, where: 
$\Xx' = \Xx \times \set{0, \ldots, |\Qq| - 1 }$ indicates the state set, $\Xx'_0 = \Xx_0 \times \set{ q_0 \mid q_0 \in \Qq_{0}}$ indicates the initial set of states, and the state transition relation $f'$ is defined as 
     $f'((x,q_i)) = \big\{ (x', q_j) \mid q_j \in \delta(q_i, \Ll(x)) \big\}$.
To verify whether a given system satisfies a desired LTL property, we make use of an ICC on the product $\Sys \otimes \Aa$.

\begin{definition}[ICC for LTL]
\label{def:icc_ltl}
   Consider a system $\Sys = (\Xx, \Xx_0, f)$ and NBA $\Aa= (\Sigma,\Qq, \Qq_0, \delta, \Qq_F)$ that represents the complement of an LTL formula $\phi$. A set of functions $\Tt_i: \Xx \times \Qq \times \Xx \times \Qq \to \R$,$\forall i \in \set{0,\ldots,k}$, is an ICC for $\Sys \otimes \Aa$ if $\exists \eta \in \R_{ > 0}$ such that $\forall x, y, y' \in \Xx$, $\forall x' \in f(x)$, and $\forall x_0 \in \Xx_0$, $\forall n,p \in \Qq$, $\forall n' \in \delta(n, \Ll(x))$,  $\forall q_0 \in \Qq_0$ and $\forall r,r' \in \Qq_F$: 
    \begin{align}
        \label{eq:icc_ltl1}
        & \Tt_0\big( (x, n), (x', n') \big) \geq 0, \\
        \label{eq:icc_ltl2}
        & \Big( \Tt_i \big( (x', n'), (y, p) \big) \geq 0 \Big) {\implies} \Big( \Tt_{i+1} \big((x, n), (y, p) \big) \geq 0 \Big) \nonumber\\
        & \qquad \forall\ 0\leq i < k,\\
        \label{eq:icc_ltl3}
        & \Big(\Tt_k \big( (x', n'), (y, p) \big) \geq 0 \Big) {\implies} \Big( \Tt_k \big((x, n), (y, p) \big) \geq 0 \Big), \\
        &\Big(\Tt_k \big((x_0,q_0),(y,r) \big) \geq 0 \Big) \wedge \Big( \Tt_k( (y,r), (y',r') ) \geq 0 \Big) \implies \nonumber \\
        \label{eq:icc_ltl4}
        & \qquad \Big(\Tt_k \big( (x_0,q_0),(y',r') \big) \leq \Tt_k \big( (x_0,q_0),(y,r) \big) - \eta \Big).  
    \end{align}
\end{definition}

Now, we provide the next result on the verification of LTL specifications using ICC on $\Sys \otimes \Aa$.
\begin{theorem}[ICCs verify LTL]
\label{thm:icc_ltl}
   Consider a system $\Sys$ and an LTL formula $\phi$. Let NBA $\Aa$ represent the complement of the specification, \textit{i.e}, $\neg \phi$. The existence of a set of functions $\Tt_i: \Xx \times \Qq \times \Xx \times \Qq \to \R$, $\forall i \in \set{0,\ldots,k}$, satisfying conditions~\eqref{eq:icc_ltl1}-\eqref{eq:icc_ltl4} implies that $\Sys \models_{\Ll} \phi$. 
\end{theorem}
\begin{proof}
    Observe that an ICC satisfying conditions~\eqref{eq:icc_ltl1}-\eqref{eq:icc_ltl4} mirrors an ICC for persistence for the product $\Sys \otimes \Aa$.
    From Theorem~\ref{thm:icc_pers}, we observe that the product system visits the accepting states only finitely often. Thus, we infer that no trace of the system is in the language of the NBA $\Aa$, thereby verifying the LTL specification.
\end{proof}

The advantage of ICCs over CCs is that ICCs enable the use of multiple simpler templates as certificates, rather than a single, more intricate template as in the case of CCs (cf. case studies). This makes the search for suitable certificates more manageable. In the next section, we introduce a computational tool for finding ICCs and discuss the benefits of ICCs from the perspective of computational complexity.

\section{Computation of ICCs}
\label{sec:synth}
This section presents an approach to synthesize ICCs using sum-of-squares (SOS)~\cite{parrilo2003semidefinite} programming when the relevant sets of the system $\Sys = (\Xx, \Xx_0, f)$ are semi-algebraic sets which are subsets of $\R^n$ (\textit{i.e.}, $\Xx \subseteq \R^n$) and the transition function $f:\Xx \to \Xx$ is polynomial.
A set $Y \subseteq \R^n$ is semi-algebraic if it can be defined with the help of a vector of polynomial inequalities $h(x)$ as $Y = \{ x \mid h(x) \geq 0 \}$, where the inequalities are element-wise.
Moreover, SOS handles optimization problems when the constraints are written as a conjunction. 
However, an ICC as defined in Definitions \ref{def:icc_safe}, \ref{def:icc_pers} and \ref{def:icc_ltl} requires the satisfaction of logical implications, which cannot be checked using the SOS approach.
Therefore, we strengthen the conditions based on implication into those that are compatible with SOS optimization through the S-procedure \cite{yakubovich1971s}.
For example, condition \eqref{eq:cc_safe2} can be rewritten as $\Tt(x, y) - \gamma \Tt(x', y) \geq 0$, where $\gamma \in \mathbb{R}_{>0}$.
Remark that if one satisfies this inequality, then condition \eqref{eq:cc_safe2} is satisfied.
Lastly, to find ICCs, we first fix the template to be a linear combination of user-defined basis functions
$   \Tt(x,y) = \mathbf{c}^T\mathbf{p}(x,y) = \sum_{m = 1}^{n} c_m p_m(x,y)$,
where functions $p_m$ are monomials over state variables $x$ and $y$, and $c_1, \ldots, c_n$ are real coefficients.
We now describe the relevant SOS conditions for verifying safety, persistence, and LTL specifications.

\subsection{Safety}
\label{subsec:find_safety}
To adopt an SOS approach to find ICCs as in Definition~\ref{def:icc_safe}, we consider the sets $\Xx$, $\Xx_0$, and $\Xx_{u}$ to be semi-algebraic sets defined with the help of vectors of polynomial inequalities $g_{A}(x)$, $g_{0}(x)$, and $g_{u}(x)$, respectively.
As these sets are semi-algebraic, the sets $\Xx \times \Xx$ and $\Xx_0 \times \Xx_{u}$ are semi-algebraic
with the corresponding vectors $g_{B}$ and $g_{C}$, respectively.
Then the search for an ICC as in Definition~\ref{def:icc_safe} reduces to showing that the following polynomials are sum-of-squares:
\begin{align}
    \label{eq:sos_icc_safe1} 
    & \Tt_0(x, x') - \lambda_{A}^T(x)g_{A}(x), \\
    \label{eq:sos_icc_safe2}
    &\Tt_{i+1}(x, y) - \gamma_i\Tt_i(x', y)-\lambda_{B,i}^T(x,y)g_{B}(x,y) \nonumber\\
    &\qquad \forall i \in \{0,\ldots,(k-1) \},\\
    \label{eq:sos_icc_safe3} 
    &\Tt_{k}(x, y) - \gamma_k\Tt_k(x', y)-\lambda_{B,k}^T(x,y)g_{B}(x,y), \text{ and }\\
    \label{eq:sos_icc_safe4}
    &-\eta - \Tt_i(x_0, x_u) - \lambda_{C,i}^T(x_0,x_u)g_{C}(x_0,x_u) \nonumber\\
    &\qquad \forall i \in \{0,\ldots,k \},
\end{align}
where $x' = f(x)$, $\eta, \gamma_i \in\R_{ > 0}$ and the multipliers $\lambda_{A}$, $\lambda_{B,i}$, $\lambda_{C,i}$, are SOS polynomials over the state variable $x$, the state variables $x,y$, and the state variables $x_0,x_u$ over the sets $\Xx$, $\Xx \times \Xx$, and $\Xx_0 \times \Xx_{u}$, respectively.

\subsection{Persistence}
\label{subsec:find_persistence}
To find ICCs as in Definition~\ref{def:icc_pers}, we consider the sets $\Xx$, $\Xx_0$, and $\Xx_{VF}$ to be semi-algebraic sets defined with the help of vectors of polynomial inequalities $g_{A}(x)$, $g_{0}(x)$, and $g_{VF}(x)$, respectively.
As these sets are semi-algebraic, the sets $\Xx \times \Xx$ and $\Xx_0 \times \Xx_{VF} \times \Xx_{VF}$ are semi-algebraic
with the corresponding vectors $g_{B}$ and $g_{C}$, respectively.
Then the search for an ICC as in Definition~\ref{def:icc_pers} reduces to showing that the following polynomials are sum-of-squares:
\begin{align}
    \label{eq:sos_icc_pers1} 
    & \Tt_0(x, x') - \lambda_{A}^T(x)g_{A}(x), \\
    \label{eq:sos_icc_pers2}
    &\Tt_{i+1}(x, y) - \gamma_i\Tt_i(x', y)-\lambda_{B,i}^T(x,y)g_{B}(x,y) \nonumber\\
    &\qquad \forall i \in \{0,\ldots,k \},\\
    \label{eq:sos_icc_pers3} 
    &\Tt_{k}(x, y) - \gamma_k\Tt_k(x', y)-\lambda_{B,k}^T(x,y)g_{B}(x,y), \text{ and }\\
    \label{eq:sos_icc_pers4}
    & (1-\rho_1)\Tt_k(x, y) - \eta - \Tt_k(x,y') - \rho_2\Tt_k(y,y')  \nonumber \\
    & \quad - \lambda_{C}^T(x,y,y')g_{C}(x,y,y'),
\end{align}
where $x' = f(x)$,  $\eta, \gamma_i, \rho_1, \rho_2 \in\R_{ > 0}$, and multipliers $\lambda_{A}$, $\lambda_{B,i}$, $\lambda_{C}$, are sum-of-squares over the state variable $x$, the state variables $x,y$, and the state variables $x,y,y'$ over the sets $\Xx$, $\Xx \times \Xx$, and $\Xx_0 \times \Xx_{VF} \times \Xx_{VF}$, respectively.

\subsection{LTL Specifications}
\label{subsec:find_ltl}
For a given LTL specification, the relevant NBA has finitely many letters $\sigma \in \Sigma$. Without loss of generality, the set $\Xx$ can be partitioned into finitely many partitions $\Xx_{\sigma_1}, \ldots, \Xx_{\sigma_p}$, where $\forall x \in \Xx_{\sigma_m}$, we have $\Ll(x) = \sigma_m$.
Given an element $\sigma_m \in \Sigma$, we can uniquely characterize the relation $\delta_{\sigma_i}$ as $(q'_i, q_i) \in \delta_{\sigma_i}$ if and only if  $q'_i \in \delta(q_i, \sigma_i)$.
We assume that the sets $\Xx$, $\Xx_{0}$, and $\Xx_{\sigma_m}$ $\forall \sigma_m$, are semi-algebraic and characterized by polynomial vectors of inequalities $g(x)$, $ g_0(x)$, and $g_{\sigma_m, A}(x)$, respectively.
Similarly, we consider polynomial vectors of inequalities $g_{\sigma_m, B}(x, y)$ over the product space $\Xx_{\sigma_m} \times \Xx$, and $g_{C}(x_0, y, y')$ over $\Xx_0 \times \Xx \times \Xx$. 
Now, we can reduce the search for an ICC to showing that the following polynomials are SOS $\forall x, y, y' \in \Xx$, $\forall x_0 \in \Xx_0$, $\forall n, p \in \Qq$, $\forall q_0 \in \Qq_0$, $\forall r, r' \in {\Qq_F}$, and $\forall \sigma_m \in \Sigma$, such that $n' \in \delta_{\sigma_m} (n)$:
\begin{align}
\label{eq:sos_icc_ltl1}
&\Tt_0^{(n,n')}(x, x') - \lambda^T_{\sigma_m, A}(x)g_{\sigma_m, A}(x), \\
\label{eq:sos_icc_ltl2}
& \Tt_{i+1}^{(n,p)}(x, y) - \gamma_i\Tt_i^{(n',p)}(x',y) \nonumber \\
&\quad -\lambda^T_{\sigma_m, B, i}(x, y) g_{\sigma_m, B}(x, y) \quad \forall\ 0\leq i < k,\\ 
\label{eq:sos_icc_ltl3}
&\Tt_{k}^{(n,p)}(x, y) - \gamma_k\Tt_k^{(n',p)}(x',y) \nonumber \\
&\quad -\lambda^T_{\sigma_m, B, k}(x, y) g_{\sigma_m, B}(x, y), \text{ and}  \\
\label{eq:sos_icc_ltl4}
& (1-\rho_1) \Tt_k^{(q_0,r)} (x_0,y) - \eta - \Tt_k^{(q_0,r')}(x_0,y') \nonumber\\
&\quad - \rho_2 \Tt_k^{(r,r')}(y,y') - \lambda^T_{C}(x_0, y, y') g_{C}(x_0, y, y'),   
\end{align}
where $\lambda_{\sigma_m, A}$, $\lambda_{\sigma_m, B, i}$, and $\lambda_{\sigma_m, C}$ are SOS polynomials of appropriate dimensions over their respective regions, and $\rho_1, \rho_2, \gamma_i, \eta \in \R_{> 0}$.

A common tool to search for SOS polynomial certificates is to use solvers such as~\cite{wang2021tssos}.
The complexity of determining whether the equations are SOS is $O\big(k \binom{2n+d}{d} ^2 \big)$~\cite{parrilo2003semidefinite}, when searching for ICCs for safety or persistence, where $k$ indicates the number of functions of the ICC, $n$ is the dimension of the state set, and $2d$ is the degree of the polynomial.
The complexity of verifying LTL specifications is $O\big( k |\Qq|^2 \binom{2n+d}{d} ^2 \big)$, where $|\Qq|$ is the number of states in the automata~\cite{murali2024closure}.
By allowing lower degrees of functions to act as certificates in our interpolation formulation, we suffer linearly in $k$ while alleviating the polynomial complexity of the degree $2d$.
Thus, one may use existing techniques to reduce the computational burden on the search for such certificates.

In the next section, we construct a scenario program to search for ICCs to verify $\omega$-regular properties.

\section{Data-Driven Computation of ICCs}
\label{sec:datadriven}

To adopt a data-driven technique, we construct a scenario program (SP) by collecting a finite set of samples $S = \{ \bar{x}_1, \dots, \bar{x}_N \}$ from our state space as follows. 
We first pick a discretization parameter $\epsilon \in \R_{> 0}$ and cover the state set $\Xx$ by finitely many sets $\Xx_i$ such that $\Xx \subseteq \bigcup_{1 \leq i \leq N} \Xx_i$. 
We select $\bar{x}_i$ from $\Xx_i$ so that for every $x \in \Xx$, we have $\| \bar{x}_i - x \| \leq \epsilon $ for some $1 \leq i \leq N$. 
For example, one can divide $\Xx$ into finitely many hypercubes, each with sides $2 \epsilon$ and select the centers of these hypercubes as sample points. 
We assume that the labeling function $\Ll$ does not change within the hypercubes (i.e., $\Ll(x) = \Ll(\bar{x}_i)\ \forall x \in \Xx_i$).
We use $\bar{\Xx}$ to denote $\Xx \cap S$ and $\bar{\Xx}_0$ to denote $\Xx_0 \cap S$.
Now, we present an SP to search for an ICC of an LTL specification using a robustness parameter $\delta$.

\begin{align}
\label{eq:dd_sp}
\text{SP:} 
\begin{cases}
    \min \limits_{\delta, \eta} & \delta + \eta,\\
    \mathsf{s.t.} & g_0^{(n,n')}(\delta, \bar{x}) \geq 0,\\
    & g_{i}^{(n,n',p)}(\delta, \bar{x}, \bar{y}) \geq 0 \ \forall\ 0\leq i < k,\\
    & g_{k}^{(n,n',p)}(\delta, \bar{x}, \bar{y}) \geq 0,\\
    & g^{(q_0,r,r')}(\delta,\eta, \bar{x}_0, \bar{y}, \bar{y}') \geq 0,\\
    & \bar{x}, \bar{y}, \bar{y}' \in \bar{\Xx}, \bar{x}_0 \in \bar{\Xx}_0, \bar{x}' \in f(\bar{x}),\\
    & \rho_1, \rho_2, \gamma_i, \eta \in \R_{> 0}, \delta \in \R,\\
    &n,n',p \in \Qq, q_0 \in \Qq_0, r,r' \in \Qq_F,
\end{cases}
\end{align}

where, 
\begin{align}
\label{eq:dd_icc_ltl1}
g_0^{(n,n')}(\delta, \bar{x}) &= \Tt_0^{(n,n')}(\bar{x}, \bar{x}') + \delta, \\
\label{eq:dd_icc_ltl2}
g_{i}^{(n,n',p)}(\delta, \bar{x}, \bar{y}) &= \Tt_{i+1}^{(n,p)}(\bar{x}, \bar{y}) - \gamma_i\Tt_i^{(n',p)}(\bar{x}',\bar{y}) + \delta, \\ 
\label{eq:dd_icc_ltl3}
g_{k}^{(n,n',p)}(\delta, \bar{x}, \bar{y}) &= \Tt_{k}^{(n,p)}(\bar{x}, \bar{y}) - \gamma_k\Tt_k^{(n',p)}(\bar{x}',\bar{y}) + \delta, \\
\label{eq:dd_icc_ltl4}
g^{(q_0,r,r')}(\delta,\eta, \bar{x}_0, \bar{y}, \bar{y}') &= (1-\rho_1) \Tt_k^{(q_0,r)} (\bar{x}_0,\bar{y}) - \eta + \delta \nonumber\\
& - \Tt_k^{(q_0,r')}(\bar{x}_0,\bar{y}') - \rho_2 \Tt_k^{(r,r')}(\bar{y},\bar{y}').     
\end{align}

In order to prove the correctness of the SP, we rely on the following assumption.
\begin{assumption}
    We assume functions $g^{(q_0,r,r')}$, $g_0^{(n,n')}$, and $g_i^{(n,n',p)}$ to be Lipschitz-continuous in $\Xx$ with Lipschitz-constants $L^{(q_0,r,r')}$, $L_0^{(n,n')}$, and $L_i^{(n,n',p)}$ for all $0 \leq i \leq k$, respectively, with respect to the $\infty$-norm 
    The maximum of these is denoted by $L'$.
\end{assumption}

The following theorem indicates when a solution to the SP is a valid ICC to verify an LTL specification.
\begin{theorem}
\label{thm:dd_sp}
    Consider a system $\Sys = (\Xx, \Xx_0, f)$ and an LTL formula $\phi$. Let NBA $\Aa$ represent the complement of the specification, \textit{i.e}, $\neg \phi$. Let $\Xx$ be partitioned into $N$ finite covers and the SP be constructed by selecting representative elements from each cover.
    Let the optimal decision variables of the SP be $[\delta^*;\eta^*]$.
    If $L' \epsilon + \delta^* \leq 0$, then the set of functions $\Tt_i: \Xx \times \Qq \times \Xx \times \Qq \to \R$, $\forall i \in \set{0,\ldots,k}$, satisfying condition~\eqref{eq:dd_sp} implies that $\Sys \models_{\Ll} \phi$. 
\end{theorem}

\begin{proof}
    From equation~\eqref{eq:dd_icc_ltl4}, one can directly check the satisfaction of condition~\eqref{eq:icc_ltl4}.
    We need to show that conditions~\eqref{eq:icc_ltl1}-\eqref{eq:icc_ltl3} hold.
    For every point $x \in \Xx$, there exists $\bar{x}_j \in S$ such that $\| \bar{x}_j - x \| \leq \epsilon$. For condition~\eqref{eq:icc_ltl1}, since $g_0^{(n,n')}$ are Lipschitz-continuous, we have $\| g_0^{(n,n')}(\delta^*, \bar{x}_j) - g_0^{(n,n')}(\delta^*, x) \| \leq L' \| \bar{x}_j - x\|$.
    Simplifying this inequality, we get $g_0^{(n,n')}(\delta^*, x) \geq g_0^{(n,n')}(\delta^*, \bar{x}_j) - L' \epsilon$.
    $g_0^{(n,n')}(\delta^*, \bar{x}_j) \geq 0$ from the SP.
    If $L' \epsilon \leq -\delta$, we have $g_0^{(n,n')}(\delta^*, x) \geq 0$ for all $x \in \Xx$. Similarly, for conditions~\eqref{eq:icc_ltl2}-\eqref{eq:icc_ltl3}, $\| g_i^{(n,n',p)}(\delta^*, \bar{x}_j, \bar{y}_j) - g_i^{(n,n',p)}(\delta^*, x, y) \| \leq L' \| [\bar{x}_j; \bar{y}_j] - [x;y]\|$. 
    It follows from simplifying these inequalities that if $L' \epsilon + \delta^* \leq 0$, $g_i^{(n,n',p)}(\delta^*, x, y) \geq 0$ for all $x, y \in \Xx$ and $0 \leq i \leq k$. 
    This proves that the set of functions $\Tt_i$ is a valid ICC to verify the LTL specifications.  
\end{proof}

Note that the above condition requires $\delta^* \leq 0$. We can now use linear programming solvers such as Gurobi~\cite{gurobi} to search for the functions per the SP in~\eqref{eq:dd_sp}. 
Assuming that the bilinear variables are fixed, one can then adopt the method discussed in~\cite{wood1996estimation}
to estimate the Lipschitz-constants of all the above functions and check for the condition per Theorem~\ref{thm:dd_sp}.

\section{Case Studies}
\label{sec:case_studies}
This section covers the case studies used to demonstrate the utility of ICCs for persistence and LTL specification. Given a system with state $x = x(t)$, we denote the next state using $x' = x(t+1) = f(x)$. The simulations were conducted on a Windows 11 device equipped with an AMD Ryzen 9 4900HS Processor and 16GB of RAM. We used TSSOS \cite{wang2021tssos} in Julia to implement the SOS formulations provided in Section \ref{sec:synth}.
The monomials and corresponding coefficients for the ICC functions can be found in Appendix~\ref{appendix:casestudies}.

\subsection{Persistence}
\label{subsec:persistence_casestudy}

\begin{figure}[!t]
    \centering
    \epsfig{file=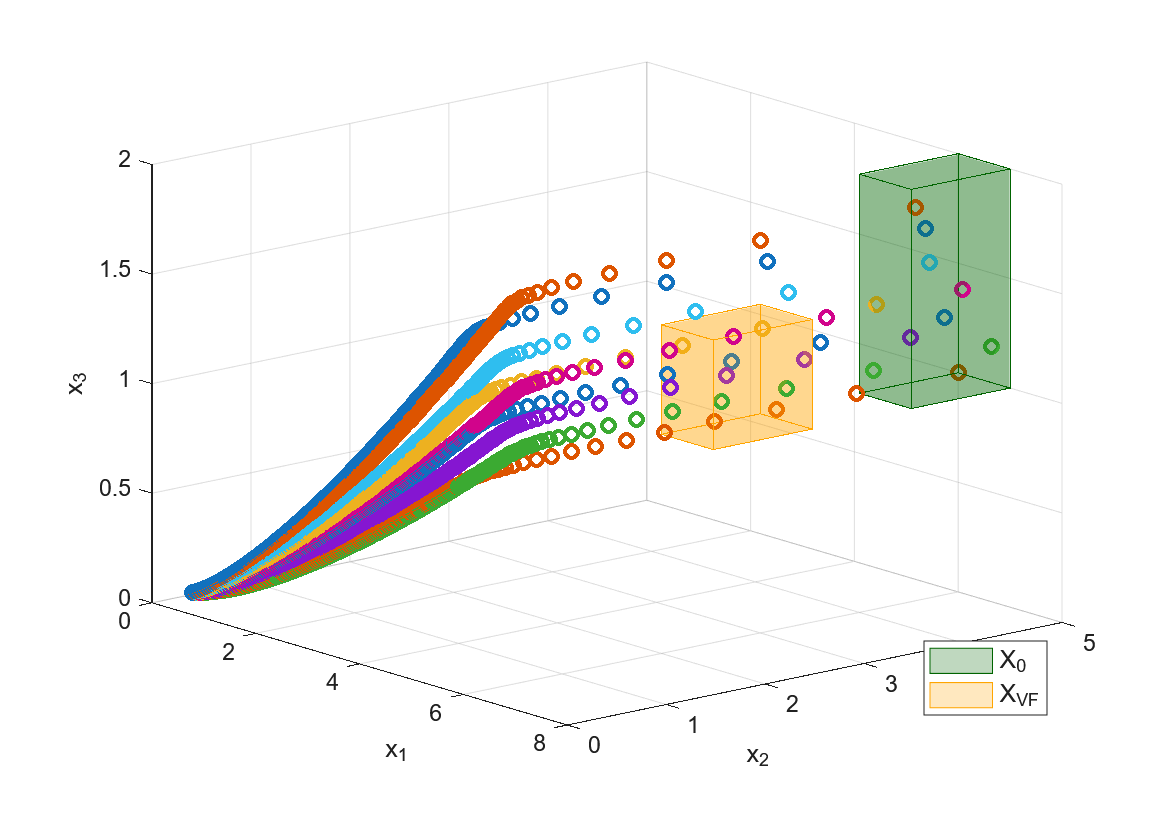, width=0.5\textwidth, keepaspectratio}
  \caption{Sample state trajectories of Lotka-Volterra dynamic model. Set $\Xx_{VF}$ is only visited finitely often.}
  \label{fig:casestudy_pers}
\end{figure} 

We experimentally demonstrate the utility of  ICCs for verifying persistence in a 3-state Lotka-Volterra type prey-predator model with states $x = [x_1, x_2, x_3]^T$.

The dynamics of our model is given by the following difference equations:
\begin{align}
&\begin{cases}
    x_1' = x_1 + T(\alpha x_2 - \beta_1 x_1),\\
    x_2' = x_2 + T(\theta x_1 - \theta x_2^2 - \psi x_2 x_3 - \beta_2 x_2),\\
    x_3' = x_3 + T(\theta x_2 - \beta_3 x_3),
\end{cases}
\end{align}
where $T = 0.01s$, $\alpha = 1.6$, $\beta_1 = 0.38$, $\beta_2 = 0.06$, $\beta_3 = 0.5$, $\theta = 0.3$, and $\psi = 20$. 
Observe that there is a jump in the trajectory of the state $x_1$ from the initial set $\Xx_0$ to the next step.
The state set, initial set, and finitely visited state set are given by $\Xx = [0,7]\times[0,5]\times[0,2]$, $\Xx_0 = [6,7]\times[4,5]\times[1,2]$, and $\Xx_{VF} = [6,7]\times[2,3]\times[1,1.5]$, respectively. 
We vary $k$ from $0$ up to $k_{max} = 2$ and the degree of the polynomial template from $1$ up to $degree_{max} = 5$, above which the SOS fails to compute due to device memory constraints. 
Note that a result with $k = 0$ is a standard closure certificate.
We reformulate conditions \eqref{eq:icc_pers1}-\eqref{eq:icc_pers4} as SOS optimization problem as described in Section \ref{subsec:find_persistence} with $\eta = 0.01$, $\gamma_i = \rho_1 = \rho_2 = 1$. 
We were unable to find a standard closure certificate. 
However, we found an ICC using a degree four polynomial template in two state variables $x,y$ and $k = 2$.  
Figure \ref{fig:casestudy_pers} shows sample runs of the dynamical system that confirm $\Xx_{VF}$ is visited only finitely often. 
It took 133s to find the ICC and 492s to try up to degree five CCs.

\subsection{LTL Specification}
\label{subsec:ltl_casestudy}

As a second case study, we exhibit the utility of  ICCs to verify LTL specifications in a heat transfer model that determines the temperatures $x = [x_1, x_2]^T$ in two adjacent rooms of a house. The dynamic model is given by the difference equations:
\begin{align}
&\begin{cases}
    x_1' = x_1 + \alpha - \beta (x_1 - \kappa) - \nu (x_1 - x_2),\\
    x_2' = x_2 + \zeta - \xi (x_2 - \kappa) - \theta (x_2 - x_1),
\end{cases}
\end{align}
where constants $\alpha = -4.9$, and $\zeta = -3.23$ are the heat contributions of a cooler in each room, $\beta = 0.49$, $\nu = 0.5$, $\xi = 0.323$, and $\theta = 0.667$ represent the conduction factors, and $\kappa = 10^{\circ} C$ is the outside temperature. The
set $\Xx = [0, 12] \times [0, 10] \in \R^{2}$ indicates the temperature of the two rooms, and $\Xx_0 = [10,12] \times [8,10]$ indicates the initial states.

Let the LTL formula to be verified be $a\ \mathsf{U}\ \mathsf{G} c$.
This property requires that a system that starts from a state labeled with an atomic proposition $a$ stays there until it reaches and always stays in states labeled with an atomic proposition $c$. The complement of this specification is denoted by an NBA $\Aa$ in Figure~\ref{fig:ltl_casestudy_nba}. We introduce two additional labels: $b$ for when the state is labeled with both $a$ and $c$, and $d$ for when the state is labeled with neither. Explicitly, the labeling map $\Ll : \Xx \rightarrow \Sigma$ is given by:
\[
\Ll(x) = 
\begin{cases}
    a \quad \text{if } x \in [0, 12] \times [7, 10],\\
    c \quad \text{if } x \in [0, 10] \times [0, 8],\\
    b \quad \text{if } x \in [0, 10] \times [7, 8],\\
    d \quad \text{otherwise}.
\end{cases}
\]

We reformulate LTL conditions as an SOS optimization problem as described in Section \ref{subsec:find_ltl} with $\eta = 0.1$, $\rho_1 = 0.5$, $\gamma_i = \rho_2 = 1$. 
We vary $k$ from $0$ up to $k_{max} = 3$ and the degree of the polynomial template from $1$ up to $degree_{max} = 6$.
We were unsuccessful in finding a standard closure certificate up to degree six. 
However, we found the piecewise components of the ICC using a degree three polynomial template and $k = 2$.
It took 33s to find the ICC for the product system and 52s to try up to degree six CCs.

\begin{figure}[t]
    \centering
    \begin{tikzpicture}[node distance =1.5cm]
    \node[initial, accepting, state, draw, initial text =,fill=blue!10!white] (1) at (0,0) {$q_1$};
    \node[,state, fill=blue!10!white,] (2) at (2,0) {$q_2$};
    \node[state, fill = blue!10!white] (3) at (2, -2) {$q_3$};
    \node[accepting,state, fill=blue!10!white,] (4) at (4,0) {$q_4$};
    \path[->]
    (1) edge[loop above] node[above]{$a$} (1)
    (1) edge[bend left] node[above]{$b$} (2)
    (1) edge[] node[below]{$c$} (3)
    (1) edge[bend left = 2.3cm] node[above]{$d$} (4)
    (2) edge[bend left] node[above]{$a$} (1)
    (2) edge[loop above] node[right=0.1cm]{$b$} (2)
    (2) edge[] node[left]{$c$} (3)
    (2) edge[] node[above]{$d$} (4)
    (3) edge[loop left] node[left]{$b,c$} (3)
    (3) edge[] node[right]{$a,d$} (4)
    (4) edge[loop right] node[right]{$\top$} (4);   
    \end{tikzpicture}
    \caption{An NBA $\Aa$ for the two-room temperature case study from Section~\ref{subsec:ltl_casestudy}. 
    The automata represents the LTL formula $\neg( a\ \mathsf{U}\ \mathsf{G} c)$.
    $\top$ indicates any letter in the alphabet.}
    \label{fig:ltl_casestudy_nba}
\end{figure}
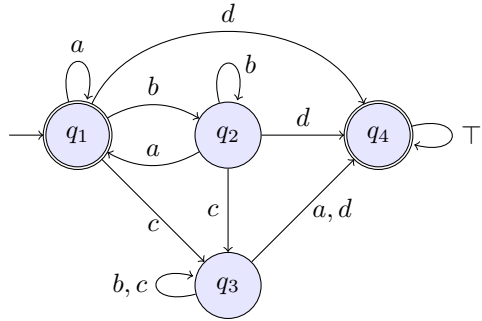

\section{Conclusion}

We proposed a notion of interpolation-inspired closure certificates (ICCs) that relaxes the conditions of a closure certificate (CC) by incrementally finding functions that together verify safety, persistence, or general LTL specifications for a given dynamical system.
We presented how to compute these certificates using SOS programming and scenario programs, and showed the potential computational advantage of using ICCs over standard CCs. 
In future work, we intend to explore how approaches such as $k$-induction~\cite{wahl2013k} and vector closure certificates can be used to allow a larger class of functions to act as closure certificates. We also plan to investigate the use of ICCs in synthesizing controllers.

\bibliographystyle{alpha}
\bibliography{ref.bib}

\newpage
\clearpage
\appendix

\section{Case Study Results}
\label{appendix:casestudies}

\topcaption{ICC function parameters for verifying persistence.
\label{tab:coeffs_persistence}}
\begin{supertabular}{|p{2em}|p{0.87\textwidth}|}
  \hline
  $\mathbf{p}$ &
  [$1$, $y_{3}$, $y_{2}$, $y_{1}$, $x_{3}$, $x_{2}$, $x_{1}$, $y_{3}^{2}$, $y_{2}y_{3}$, $y_{2}^{2}$, $y_{1}y_{3}$, $y_{1}y_{2}$, $y_{1}^{2}$, $x_{3}y_{3}$, $x_{3}y_{2}$, $x_{3}y_{1}$, $x_{3}^{2}$, $x_{2}y_{3}$, $x_{2}y_{2}$, $x_{2}y_{1}$, $x_{2}x_{3}$, $x_{2}^{2}$, $x_{1}y_{3}$, $x_{1}y_{2}$, $x_{1}y_{1}$, $x_{1}x_{3}$, $x_{1}x_{2}$, $x_{1}^{2}$, $y_{3}^{3}$, $y_{2}y_{3}^{2}$, $y_{2}^{2}y_{3}$, $y_{2}^{3}$, $y_{1}y_{3}^{2}$, $y_{1}y_{2}y_{3}$, $y_{1}y_{2}^{2}$, $y_{1}^{2}y_{3}$, $y_{1}^{2}y_{2}$, $y_{1}^{3}$, $x_{3}y_{3}^{2}$, $x_{3}y_{2}y_{3}$, $x_{3}y_{2}^{2}$, $x_{3}y_{1}y_{3}$, $x_{3}y_{1}y_{2}$, $x_{3}y_{1}^{2}$, $x_{3}^{2}y_{3}$, $x_{3}^{2}y_{2}$, $x_{3}^{2}y_{1}$, $x_{3}^{3}$, $x_{2}y_{3}^{2}$, $x_{2}y_{2}y_{3}$, $x_{2}y_{2}^{2}$, $x_{2}y_{1}y_{3}$, $x_{2}y_{1}y_{2}$, $x_{2}y_{1}^{2}$, $x_{2}x_{3}y_{3}$, $x_{2}x_{3}y_{2}$, $x_{2}x_{3}y_{1}$, $x_{2}x_{3}^{2}$, $x_{2}^{2}y_{3}$, $x_{2}^{2}y_{2}$, $x_{2}^{2}y_{1}$, $x_{2}^{2}x_{3}$, $x_{2}^{3}$, $x_{1}y_{3}^{2}$, $x_{1}y_{2}y_{3}$, $x_{1}y_{2}^{2}$, $x_{1}y_{1}y_{3}$, $x_{1}y_{1}y_{2}$, $x_{1}y_{1}^{2}$, $x_{1}x_{3}y_{3}$, $x_{1}x_{3}y_{2}$, $x_{1}x_{3}y_{1}$, $x_{1}x_{3}^{2}$, $x_{1}x_{2}y_{3}$, $x_{1}x_{2}y_{2}$, $x_{1}x_{2}y_{1}$, $x_{1}x_{2}x_{3}$, $x_{1}x_{2}^{2}$, $x_{1}^{2}y_{3}$, $x_{1}^{2}y_{2}$, $x_{1}^{2}y_{1}$, $x_{1}^{2}x_{3}$, $x_{1}^{2}x_{2}$, $x_{1}^{3}$, $y_{3}^{4}$, $y_{2}y_{3}^{3}$, $y_{2}^{2}y_{3}^{2}$, $y_{2}^{3}y_{3}$, $y_{2}^{4}$, $y_{1}y_{3}^{3}$, $y_{1}y_{2}y_{3}^{2}$, $y_{1}y_{2}^{2}y_{3}$, $y_{1}y_{2}^{3}$, $y_{1}^{2}y_{3}^{2}$, $y_{1}^{2}y_{2}y_{3}$, $y_{1}^{2}y_{2}^{2}$, $y_{1}^{3}y_{3}$, $y_{1}^{3}y_{2}$, $y_{1}^{4}$, $x_{3}y_{3}^{3}$, $x_{3}y_{2}y_{3}^{2}$, $x_{3}y_{2}^{2}y_{3}$, $x_{3}y_{2}^{3}$, $x_{3}y_{1}y_{3}^{2}$, $x_{3}y_{1}y_{2}y_{3}$, $x_{3}y_{1}y_{2}^{2}$, $x_{3}y_{1}^{2}y_{3}$, $x_{3}y_{1}^{2}y_{2}$, $x_{3}y_{1}^{3}$, $x_{3}^{2}y_{3}^{2}$, $x_{3}^{2}y_{2}y_{3}$, $x_{3}^{2}y_{2}^{2}$, $x_{3}^{2}y_{1}y_{3}$, $x_{3}^{2}y_{1}y_{2}$, $x_{3}^{2}y_{1}^{2}$, $x_{3}^{3}y_{3}$, $x_{3}^{3}y_{2}$, $x_{3}^{3}y_{1}$, $x_{3}^{4}$, $x_{2}y_{3}^{3}$, $x_{2}y_{2}y_{3}^{2}$, $x_{2}y_{2}^{2}y_{3}$, $x_{2}y_{2}^{3}$, $x_{2}y_{1}y_{3}^{2}$, $x_{2}y_{1}y_{2}y_{3}$, $x_{2}y_{1}y_{2}^{2}$, $x_{2}y_{1}^{2}y_{3}$, $x_{2}y_{1}^{2}y_{2}$, $x_{2}y_{1}^{3}$, $x_{2}x_{3}y_{3}^{2}$, $x_{2}x_{3}y_{2}y_{3}$, $x_{2}x_{3}y_{2}^{2}$, $x_{2}x_{3}y_{1}y_{3}$, $x_{2}x_{3}y_{1}y_{2}$, $x_{2}x_{3}y_{1}^{2}$, $x_{2}x_{3}^{2}y_{3}$, $x_{2}x_{3}^{2}y_{2}$, $x_{2}x_{3}^{2}y_{1}$, $x_{2}x_{3}^{3}$, $x_{2}^{2}y_{3}^{2}$, $x_{2}^{2}y_{2}y_{3}$, $x_{2}^{2}y_{2}^{2}$, $x_{2}^{2}y_{1}y_{3}$, $x_{2}^{2}y_{1}y_{2}$, $x_{2}^{2}y_{1}^{2}$, $x_{2}^{2}x_{3}y_{3}$, $x_{2}^{2}x_{3}y_{2}$, $x_{2}^{2}x_{3}y_{1}$, $x_{2}^{2}x_{3}^{2}$, $x_{2}^{3}y_{3}$, $x_{2}^{3}y_{2}$, $x_{2}^{3}y_{1}$, $x_{2}^{3}x_{3}$, $x_{2}^{4}$, $x_{1}y_{3}^{3}$, $x_{1}y_{2}y_{3}^{2}$, $x_{1}y_{2}^{2}y_{3}$, $x_{1}y_{2}^{3}$, $x_{1}y_{1}y_{3}^{2}$, $x_{1}y_{1}y_{2}y_{3}$, $x_{1}y_{1}y_{2}^{2}$, $x_{1}y_{1}^{2}y_{3}$, $x_{1}y_{1}^{2}y_{2}$, $x_{1}y_{1}^{3}$, $x_{1}x_{3}y_{3}^{2}$, $x_{1}x_{3}y_{2}y_{3}$, $x_{1}x_{3}y_{2}^{2}$, $x_{1}x_{3}y_{1}y_{3}$, $x_{1}x_{3}y_{1}y_{2}$, $x_{1}x_{3}y_{1}^{2}$, $x_{1}x_{3}^{2}y_{3}$, $x_{1}x_{3}^{2}y_{2}$, $x_{1}x_{3}^{2}y_{1}$, $x_{1}x_{3}^{3}$, $x_{1}x_{2}y_{3}^{2}$, $x_{1}x_{2}y_{2}y_{3}$, $x_{1}x_{2}y_{2}^{2}$, $x_{1}x_{2}y_{1}y_{3}$, $x_{1}x_{2}y_{1}y_{2}$, $x_{1}x_{2}y_{1}^{2}$, $x_{1}x_{2}x_{3}y_{3}$, $x_{1}x_{2}x_{3}y_{2}$, $x_{1}x_{2}x_{3}y_{1}$, $x_{1}x_{2}x_{3}^{2}$, $x_{1}x_{2}^{2}y_{3}$, $x_{1}x_{2}^{2}y_{2}$, $x_{1}x_{2}^{2}y_{1}$, $x_{1}x_{2}^{2}x_{3}$, $x_{1}x_{2}^{3}$, $x_{1}^{2}y_{3}^{2}$, $x_{1}^{2}y_{2}y_{3}$, $x_{1}^{2}y_{2}^{2}$, $x_{1}^{2}y_{1}y_{3}$, $x_{1}^{2}y_{1}y_{2}$, $x_{1}^{2}y_{1}^{2}$, $x_{1}^{2}x_{3}y_{3}$, $x_{1}^{2}x_{3}y_{2}$, $x_{1}^{2}x_{3}y_{1}$, $x_{1}^{2}x_{3}^{2}$, $x_{1}^{2}x_{2}y_{3}$, $x_{1}^{2}x_{2}y_{2}$, $x_{1}^{2}x_{2}y_{1}$, $x_{1}^{2}x_{2}x_{3}$, $x_{1}^{2}x_{2}^{2}$, $x_{1}^{3}y_{3}$, $x_{1}^{3}y_{2}$, $x_{1}^{3}y_{1}$, $x_{1}^{3}x_{3}$, $x_{1}^{3}x_{2}$, $x_{1}^{4}$]$^T$
  \\
  \hline
  $\mathbf{c}_0$ & [546.774, 97.77, 101.115, 78.106, 19.979, -7.635, -66.302, -1.058, -52.48, 14.107, -118.303, -241.636, -60.813, 170.9, 36.16, -15.847, 291.543, 16.21, 385.094, -34.412, -34.105, -55.251, -25.127, -58.254, 130.89, -69.873, -4.746, -56.542, -78.507, -20.038, -124.717, -382.773, 19.032, -34.597, 7.232, -46.126, -13.919, -67.051, 56.092, 50.142, -162.988, 76.588, -12.671, -10.887, 95.613, -13.31, 11.728, -72.48, -52.317, 67.452, 362.652, 32.636, 142.832, -58.813, 37.507, 106.7, 50.803, -32.529, 50.256, 149.356, -65.567, 56.252, -247.807, -29.615, -36.317, -48.646, 18.016, 72.422, 62.051, 63.178, $-14.161$, 94.46, -2.489, 31.202, 150.245, 48.076, 48.539, -61.793, -51.814, -25.639, 84.498, -8.705, -30.213, -67.088, -304.301, -122.726, -352.426, -154.53, 22.712, -81.41, -77.668, 0.167, -41.64, -275.009, -70.695, -280.089, -92.85, -177.307, -580.874, 200.008, 28.901, 111.665, -148.26, 80.328, 84.398, -5.592, 312.252, -21.682, -113.214, 219.06, 60.984, -249.231, 90.261, -67.293, -262.267, 286.904, -82.628, -121.264, -258.626, -84.85, -28.319, 21.626, -90.973, -57.624, -35.376, 206.229, -6.917, 76.479, -135.192, 96.857, 247.832, 116.085, 104.619, -55.056, -33.591, 126.115, 99.056, -90.957, -63.048, -161.603, -40.673, -93.257, 56.469, -249.913, -199.253, 263.791, 59.033, -5.576, -31.226, 61.937, 326.432, 85.057, 66.612, -139.596, -100.272, -53.704, 13.465, -38.843, 5.036, 24.587, 282.572, 94.85, 181.657, 328.28, 49.931, 91.555, -17.189, 493.447, 76.444, 108.942, 64.687, -71.689, 17.813, -135.858, -61.963, -29.947, 182.905, 101.52, 319.473, 167.892, 91.42, -51.493, 81.217, -109.47, 55.187, -245.915, 207.058, -33.263, 69.851, -271.195, -58.376, -264.747, 87.986, 171.788, 491.666, 309.536, -13.202, 106.842, -271.473, -9.653, 88.802, 139.942, -38.702, -211.218, -89.546, -191.039, 338.751, -117.011, -165.714, -577.504]$^T$\\ 
  \hline
  $\mathbf{c}_1$ & [623.663, 110.37, 94.691, 107.924, 31.128, 14.544, -53.92, 95.111, -149.301, 238.736, $-169.908$, -423.932, -1.466, 176.372, $-4.731$, -17.365, 510.305, 19.928, 206.544, -8.425, $-22.444$, 98.078, -41.067, -104.044, 56.175, -108.661, 5.835, -1.004, -57.747, 26.133, $-68.07$, -174.315, 63.063, -58.03, 22.678, $-30.964$, -0.746, -25.477, 56.904, 55.874, $-102.525$, 73.959, -19.794, -12.821, 83.158, 33.27, 32.395, -74.308, -27.341, 56.601, 177.904, 21.916, 119.09, -35.782, 18.845, 29.668, 45.243, -35.31, 36.402, -28.557, $-19.874$, 70.307, -108.357, -3.092, -45.645, $-30.97$, -8.118, 7.602, 29.435, 63.553, $-18.762$, 49.406, 25.367, 22.786, 127.706, 25.259, 52.145, -34.485, -45.344, -31.904, 52.727, -7.975, -5.17, -34.94, -175.973, $-70.644$, -203.612, -82.567, -3.129, -29.897, $-44.092$, -11.881, -15.616, -132.963, $-44.954$, -136.44, -42.805, -70.499, -272.985, 101.274, 22.264, 27.237, -65.406, 27.337, 24.242, -15.235, 149.632, -15.171, -51.971, 179.359, 41.826, -95.611, 26.464, -20.156, -105.775, 173.149, -40.004, -67.08, -124.956, -27.749, -24.292, 34.617, -32.021, -15.561, 13.905, 111.97, 1.475, 16.926, -32.294, 64.087, 74.326, 55.479, 10.513, -20.362, 8.444, 29.343, -18.226, -59.055, -29.699, -46.338, -34.412, -30.054, 5.027, -125.784, -53.507, 46.506, -81.929, -19.602, -11.092, 20.844, 120.813, 24.103, 130.904, -32.296, -49.603, -11.038, 4.161, -17.247, 5.981, 15.913, 155.667, 49.239, 81.62, 166.508, -8.755, 32.072, -34.382, 220.341, 16.59, 43.995, 0.533, -26.538, 10.77, -74.689, -26.24, 13.173, 78.228, 44.65, 121.984, 51.042, 4.954, -10.683, 22.878, -74.344, 12.585, -103.391, 75.739, -36.8, 12.823, -124.82, -33.815, -115.761, 37.09, 64.368, 201.112, 146.056, -5.036, 40.558, -115.141, -2.791, 12.669, 40.01, 14.888, -56.995, -39.778, -78.232, 173.295, -48.341, -63.121, -268.259]$^T$\\ 
  \hline
  $\mathbf{c}_2$ & [901.99, 202.133, 51.557, 290.213, 58.409, $-38.713$, -26.588, 166.025, -391.969, 667.862, -270.787, -1111.486, 184.282, -5.108, 1.075, 8.359, 2511.517, 12.664, 9.238, 3.889, $-56.089$, 246.634, 7.677, 5.762, 1.444, $-850.565$, -20.51, 90.488, -77.464, 52.971, 42.436, -35.287, 84.779, -221.439, 104.513, -8.148, -5.125, 2.2, 2.288, 1.411, 1.762, 0.06, 1.334, 2.514, 125.382, 78.399, 16.729, $-399.506$, -5.693, -1.718, -2.679, -0.427, -0.525, -1.451, -7.457, -20.373, -29.966, 110.376, -50.007, -43.822, -4.812, -77.447, 149.255, -4.064, -0.89, -1.852, 0.242, -0.095, -0.957, -21.366, -21.171, -21.639, 33.27, -6.698, 0.699, -2.593, 10.205, -56.812, -7.013, -5.501, 0.743, 24.892, 17.131, 3.864, -49.519, -33.509, -73.406, -19.597, -22.459, 24.355, 25.224, 23.63, 17.785, -19.918, -14.198, -14.739, 4.06, 3.638, -1.462, -0.005, 0.104, 0.017, 0.01, 0.313, -0.106, -0.144, -0.047, -0.305, -0.666, 120.102, 47.737, 78.733, $-32.547$, -27.033, 14.055, 1.892, -5.638, -5.119, -97.484, 0.0, -0.018, -0.003, -0.005, -0.061, 0.015, -0.034, 0.023, 0.066, 0.072, 5.56, 1.058, 1.703, -0.343, -1.626, 0.759, -10.23, -4.742, 2.565, -4.506, 25.445, 10.451, 15.459, 1.329, 0.747, 4.953, 12.106, 4.023, 1.34, -26.093, 3.894, 3.534, -6.407, 5.437, -24.317, 0.012, -0.025, -0.005, -0.003, -0.065, -0.0, -0.013, -0.004, 0.069, 0.125, -41.513, -17.306, -26.272, 9.2, 6.842, -5.143, -1.712, -0.196, -0.196, 67.736, 1.017, -0.911, -1.042, 0.529, 1.069, 0.568, 1.572, 0.877, -2.777, -6.228, 0.067, 0.654, 2.938, 2.722, 5.555, 7.79, 2.439, 4.182, -0.983, -0.644, 0.929, -0.264, -0.032, 0.818, -27.163, 0.197, 0.165, -0.107, 3.118, -0.569, -0.002, 0.005, -0.494, 6.133, -1.907, -0.983]$^T$\\
  \hline
\end{supertabular}

\topcaption{ICC function parameters for verifying LTL specification.\label{tab:coeffs_ltl}}
\begin{supertabular}{|p{2em}|p{0.87\textwidth}|} 
  \hline
  $\mathbf{p}$ &
  [$1$, $x_1$, $x_2$, $y_1$, $y_2$, $x_1^2$, $x_1x_2$, $x_1y_1$, $x_1y_2$, $x_2^2$, $x_2y_1$, $x_2y_2$, $y_1^2$, $y_1y_2$, $y_2^2$, $x_1^3$, $x_1^2x_2$, $x_1^2y_1$, $x_1^2y_2$, $x_1x_2^2$, $x_1x_2y_1$, $x_1x_2y_2$, $x_1y_1^2$, $x_1y_1y_2$, $x_1y_2^2$, $x_2^3$, $x_2^2y_1$, $x_2^2y_2$, $x_2y_1^2$, $x_2y_1y_2$, $x_2y_2^2$, $y_1^3$, $y_1^2y_2$, $y_1y_2^2$, $y_2^3$]$^T$ \\
  \hline
  $\mathbf{c}_0^{(1,1)}$ & [-6977.738, -2746.677, -1319.269, -674.295, -1239.944, -3146.478, -4851.562, -1547.883, 587.666, -2426.639, 1352.32, -1736.657, -5490.082, -687.797, -5498.282, -351.713, 173.071, 7.817, 586.948, 623.665, 390.322, 777.13, -64.992, 179.126, 13.021, 246.227, 591.146, 279.392, 201.586, 249.643, 154.141, -99.157, -87.349, -94.903, -87.917]$^T$\\ 
  \hline
  $\mathbf{c}_0^{(1,2)}$ & [-292.282, 24.085, 158.956, 9.865, -15.286, 1108.053, -40.135, 50.494, 980.63, 800.006, 852.045, 54.606, -2289.561, 17.713, -2217.997, -254.087, -314.746, 143.68, 10.946, -267.033, -118.376, -172.042, 726.589, 94.782, 663.168, -231.586, 144.891, 173.568, 847.778, 128.103, 822.231, -71.371, -54.358, -73.492, -106.753]$^T$\\ 
  \hline
  $\mathbf{c}_0^{(1,3)}$ & [1241.579, 52.418, 63.585, -20.061, -6.36, 1638.944, -46.67, 51.374, 1303.645, 1530.794, 1120.19, 53.751, -2169.755, 20.004, -2040.432, -239.218, -316.117, 139.646, -0.453, -255.957, -120.539, -188.944, 725.512, 86.942, 654.37, -241.442, 115.394, 158.75, 838.0, 122.242, 818.39, -67.124, -55.43, -75.292, -115.837]$^T$\\ 
  \hline
  $\mathbf{c}_0^{(1,4)}$ & [-7202.852, -2451.871, -3391.158, -969.081, -809.317, -4266.426, -4562.498, -1434.421, 445.642, -3340.438, 1084.573, -1742.485, -4741.756, -472.843, -4642.887, 205.271, 281.931, 136.0, 817.373, 120.712, 243.605, 541.103, -166.727, 120.607, -72.319, 40.594, 177.529, 266.478, -121.0, 150.993, -81.755, -75.196, -61.891, -69.846, -63.521]$^T$\\ 
  \hline
  $\mathbf{c}_0^{(2,1)}$ & [-6929.371, -2791.231, -1421.378, -716.081, -1279.742, -2913.371, -4831.396, -1611.09, 465.927, -2457.124, 1181.062, -1804.105, -5335.054, -736.425, -5349.332, -347.895, 131.66, 18.15, 579.695, 642.443, 379.194, 777.284, -53.618, 182.747, 23.729, 263.432, 589.543, 278.202, 206.036, 247.69, 161.458, $-98.51$, -86.804, -93.93, -87.524]$^T$\\ 
  \hline
  $\mathbf{c}_0^{(2,2)}$ & [-169.738, 37.705, 143.941, 4.044, -15.938, 1164.465, -40.578, 51.702, 855.172, 886.178, 727.483, 53.51, -2051.28, 16.369, -1982.704, -250.597, -329.816, 172.341, 1.712, -269.901, -124.064, -183.383, 736.428, 98.003, 668.563, -238.2, 145.297, 176.287, 851.837, 128.072, 826.479, -69.852, -52.858, -71.629, -105.58]$^T$\\ 
  \hline
  $\mathbf{c}_0^{(2,3)}$ & [1310.085, 57.104, 63.226, -20.12, -8.74, 1674.561, -43.12, 54.416, 1167.903, 1603.105, 988.018, 54.264, -1932.466, 19.173, -1807.163, -237.982, -328.298, 172.009, -6.929, -258.304, -125.583, -199.342, 735.457, 91.161, 660.573, -244.309, 120.976, 164.244, 843.301, 123.096, 822.758, -65.498, -53.845, -73.34, -114.604]$^T$\\ 
  \hline
  $\mathbf{c}_0^{(2,4)}$ & [-7108.847, -2460.835, -3442.892, -984.031, -835.257, -4032.946, -4501.074, -1485.285, 299.871, -3230.849, 1005.124, -1777.643, -4552.001, -499.204, -4484.624, 185.934, 266.887, 150.404, 812.267, 124.213, 226.491, 546.575, -157.556, 123.29, -62.104, 42.681, 172.776, 261.748, -116.048, 147.575, -75.315, -74.39, -61.42, -68.801, -63.036]$^T$\\ 
  \hline
  $\mathbf{c}_0^{(3,1)}$ & [-15467.353, -613.564, -549.224, -2002.658, -2149.028, -11520.872, -1954.158, -537.787, -432.616, -11592.03, -640.65, -684.485, $-7940.58$, -853.286, -7866.848, -1076.598, $-1205.007$, -334.695, 20.791, 943.674, 247.202, 157.05, 106.363, 101.064, 124.919, -27.344, 395.717, 174.76, 267.317, 213.679, 220.333, -84.223, -66.17, -69.712, -77.021]$^T$\\ 
  \hline
  $\mathbf{c}_0^{(3,2)}$ & [-2696.648, 101.736, 118.456, -15.425, $-20.741$, -1374.885, -18.947, 47.964, 155.917, -1559.032, 182.855, 39.213, -3479.037, -12.293, -3456.857, -12.121, -188.341, 191.448, 113.313, -86.926, -72.674, -125.455, 814.675, 127.546, 771.069, -155.542, 213.229, 210.818, 904.906, 148.631, 866.015, -63.495, -47.239, -62.556, -99.378]$^T$\\ 
  \hline
  $\mathbf{c}_0^{(3,3)}$ & [734.672, 102.446, 98.797, -49.185, $-44.705$, 1647.836, -59.972, 70.777, 1632.993, 1344.739, 1384.314, 65.94, -3231.261, 6.433, -3134.426, -233.562, -326.534, 118.207, -54.537, -255.757, -143.981, -196.53, 805.194, 88.322, 723.452, -244.39, 96.69, 138.786, 867.095, 118.626, 855.804, -62.548, -51.674, -68.449, -112.069]$^T$\\ 
  \hline
  $\mathbf{c}_0^{(3,4)}$ & [-25323.463, 2469.561, 2206.562, -1121.392, -1200.824, -7622.387, -3914.989, -1134.334, 2777.826, -6970.0, 2126.608, -1283.033, $-5903.2$, -398.361, -5633.622, 493.704, 351.343, 128.983, 622.491, 546.982, 312.928, 325.1, -36.382, 74.158, -30.43, 853.743, 587.933, 219.629, -16.528, 120.092, -50.238, -58.825, -47.968, -61.839, -58.12]$^T$\\ 
  \hline
  $\mathbf{c}_0^{(4,1)}$ & [-17566.077, -636.013, -597.042, -2630.3, $-2751.161$, -12980.515, -1357.71, -475.023, -470.547, -13028.849, -811.465, -724.624, $-9241.957$, -1061.474, -9166.107, -1519.85, $-1467.421$, -305.056, -46.485, 773.719, 148.113, 69.66, 227.421, 122.938, 225.988, 9.728, 424.347, 226.629, 354.338, 231.204, 314.501, -76.892, -50.758, -50.046, -71.074]$^T$\\ 
  \hline
  $\mathbf{c}_0^{(4,2)}$ & [-3224.071, 191.459, 201.657, -6.073, -17.888, -844.552, 2.682, 72.794, 130.234, -1093.157, 144.267, 53.004, -3909.04, -9.871, -3900.571, 79.059, -114.074, 260.217, 181.745, -66.798, $-96.253$, -138.929, 886.569, 139.048, 843.207, -104.432, 252.78, 239.727, 964.419, 160.278, 923.266, -58.187, -42.702, -54.984, -94.636]$^T$\\ 
  \hline
  $\mathbf{c}_0^{(4,3)}$ & [-1552.069, 214.156, 222.183, 0.858, -6.062, -397.992, 16.93, 87.038, 329.945, -254.785, 373.574, 69.419, -3790.916, -2.505, -3743.848, 26.383, -125.699, 256.435, 151.683, -122.051, -115.556, -155.454, 884.766, 133.914, 835.577, -106.842, 227.503, 218.03, 956.437, 154.174, 919.941, -53.913, -43.739, -56.412, -103.431]$^T$\\ 
  \hline
  $\mathbf{c}_0^{(4,4)}$ & [0.38, 2.055, 0.359, 4.034, -1.01, -2259.814, -112.241, 84.032, 6787.465, -1358.322, 5441.623, 84.327, -5461.483, 41.029, $-5099.861$, -67.482, -40.536, -3.918, 180.88, -20.228, 72.844, 103.77, -65.838, 35.603, -88.186, -25.626, 93.508, -10.936, -59.709, 41.8, -63.935, -44.658, -37.938, -45.626, -46.303]$^T$\\ 
  \hline
  $\mathbf{c}_1^{(1,1)}$ & [-7450.695, -1638.855, -532.965, -984.145, $-1109.648$, -1940.89, -1695.391, -620.865, -867.755, -128.604, -623.655, -323.263, $-3464.427$, -727.035, -3366.868, 85.36, 437.483, 281.729, 156.765, -73.722, 211.142, 156.795, 159.504, 161.418, 121.658, -2.902, -9.917, 127.033, 37.05, 65.094, 56.85, -67.733, -60.176, -64.592, -61.166]$^T$\\ 
  \hline
  $\mathbf{c}_1^{(1,2)}$ & [1894.096, 111.79, 96.853, 16.364, 10.536, 3201.101, -16.379, 268.71, 33.787, 3270.208, 24.0, 201.274, 713.325, -5.262, 753.718, -47.651, -25.424, 92.098, 95.832, -32.767, -24.896, -46.44, 598.465, 92.815, 575.998, 0.485, 56.367, 30.708, 405.664, 57.118, 380.519, -55.273, -40.485, -56.878, -92.164]$^T$\\ 
  \hline
  $\mathbf{c}_1^{(1,3)}$ & [3770.798, 96.099, 99.542, 6.775, 15.664, 3803.749, 10.285, 559.615, 48.421, 3595.703, 36.395, 461.837, 849.963, 1.793, 941.38, $-56.595$, -40.221, 71.064, 82.431, -40.724, -34.138, -57.088, 589.611, 85.901, 572.76, -9.081, 50.863, 16.814, 403.882, 50.887, 373.143, -51.316, -41.97, -58.904, -101.495]$^T$\\ 
  \hline
  $\mathbf{c}_1^{(1,4)}$ & [-6561.903, -3836.478, -560.448, -1159.815, -625.588, -1560.108, -1528.855, -125.107, -1044.214, -334.928, -652.849, -105.901, -2405.949, -445.003, -2267.029, 245.535, 166.496, 286.927, 141.742, -22.39, 90.256, 114.035, -9.296, 95.26, -36.871, 49.581, 38.364, 152.53, -29.427, 40.162, -7.028, -42.775, -37.071, -42.8, -39.823]$^T$\\ 
    \hline
  $\mathbf{c}_1^{(2,1)}$ & [-7410.341, -1752.622, -891.924, $-969.179$, -1107.985, -1689.173, -2103.496, -555.95, $-920.132$, -48.53, -768.867, -419.385, -3318.7, -696.401, -3233.646, 84.87, 424.845, 280.476, 155.323, -13.485, 221.798, 181.917, 161.765, 161.515, 126.487, 14.34, 13.428, 146.516, 40.508, 73.396, 60.759, -67.884, -60.113, $-64.336$, -60.885]$^T$\\ 
  \hline
  $\mathbf{c}_1^{(2,2)}$ & [2003.163, 117.463, 96.108, 20.022, 13.345, 3245.653, -6.156, 280.209, 44.47, 3335.657, 27.778, 201.803, 849.787, 8.231, 886.836, $-48.333$, -26.504, 91.799, 95.44, -33.921, $-25.521$, -46.834, 596.768, 92.245, 574.515, -0.654, 56.213, 30.329, 403.833, 56.594, 378.709, -55.291, -40.468, -56.743, -92.163]$^T$\\ 
  \hline
  $\mathbf{c}_1^{(2,3)}$ & [3864.738, 101.242, 99.863, 10.589, 18.463, 3845.681, 20.13, 566.952, 58.54, 3658.429, 39.996, 457.478, 985.06, 14.78, 1072.695, $-57.18$, -41.103, 70.934, 82.078, -41.707, -34.528, -57.235, 587.94, 85.336, 571.346, -10.026, 50.771, 16.721, 402.124, 50.346, 371.399, -51.325, -41.961, -58.766, -101.488]$^T$\\ 
  \hline
  $\mathbf{c}_1^{(2,4)}$ & [-6562.905, -3741.461, -1209.249, -1150.118, -713.878, -1199.411, -1991.563, -10.302, -1063.31, -69.487, -776.115, -114.753, $-2284.25$, -435.673, -2162.273, 235.977, 176.679, 282.717, 147.984, 26.186, 103.261, 133.183, -7.11, 99.856, -33.532, 63.676, 54.084, 169.81, -27.605, 47.83, -3.529, -42.897, -36.85, -42.532, -39.839]$^T$\\ 
  \hline
  $\mathbf{c}_1^{(3,1)}$ & [-14029.655, -139.625, -132.066, -2347.139, -2501.865, -3639.331, -945.171, -514.212, -466.832, -1881.96, -238.299, -208.875, -5538.221, -1004.449, -5445.73, 33.125, 214.271, 206.755, 116.79, -163.737, 54.012, 28.036, 225.267, 147.564, 197.931, -73.948, -28.868, 34.027, 109.245, 58.451, 110.541, -62.102, -41.611, -42.184, -57.609]$^T$\\ 
  \hline
  $\mathbf{c}_1^{(3,2)}$ & [311.46, 156.207, 140.642, 0.55, -6.43, 2710.538, -32.521, 82.339, 25.6, 2955.471, 16.371, 48.62, -385.48, -27.723, -367.971, -27.339, -6.955, 113.893, 108.145, -11.765, -21.317, -37.698, 639.375, 104.413, 611.906, 17.962, 65.857, 46.436, 445.662, 68.34, 424.305, -51.356, -36.868, -49.955, -88.708]$^T$\\ 
  \hline
  $\mathbf{c}_1^{(3,3)}$ & [2644.532, 157.259, 152.914, -5.18, -2.125, 3553.892, -6.091, 496.306, 45.258, 3381.025, 27.238, 374.956, -228.426, -19.339, -162.63, -43.994, -31.488, 82.437, 87.889, -26.37, -33.956, -51.801, 626.275, 94.728, 607.757, 1.357, 56.898, 25.435, 443.618, 60.15, 414.112, -48.051, -38.764, -52.57, -98.555]$^T$\\ 
  \hline
  $\mathbf{c}_1^{(3,4)}$ & [-24671.824, 1830.844, 2055.407, -1208.317, -1290.181, -2452.118, -1447.027, 412.954, -753.007, -2267.571, -585.973, 276.647, -3361.621, -430.587, -3189.371, 308.015, 20.593, 281.12, 98.156, 211.804, 56.064, 84.469, 33.649, 65.211, -0.126, 318.78, 135.161, 274.812, 19.193, 54.762, 39.473, -28.593, -23.701, -30.18, -28.327]$^T$\\ 
  \hline
  $\mathbf{c}_1^{(4,1)}$ & [-15935.814, -310.778, -158.314, -2995.753, -3112.366, -4059.302, -861.739, -651.799, -540.085, -1968.911, -242.135, -239.213, -6955.304, -1179.61, -6865.248, 37.725, 181.435, 207.579, 125.029, -199.342, 29.307, 6.643, 270.418, 147.156, 247.547, -144.334, -25.238, 20.419, 180.007, 73.438, 174.515, -62.154, -33.526, -31.418, -58.402]$^T$\\ 
  \hline
  $\mathbf{c}_1^{(4,2)}$ & [-366.009, 206.735, 208.47, 17.786, 4.918, 2909.079, -4.337, 75.001, 34.843, 3215.343, 52.647, 65.065, -1029.673, -0.279, -1025.148, -24.362, -5.993, 121.029, 113.002, -9.713, -30.548, -44.008, 674.108, 108.943, 645.691, 26.141, 78.557, 59.14, 495.921, 76.741, 474.026, -50.317, -35.973, -47.109, -87.59]$^T$\\ 
  \hline
  $\mathbf{c}_1^{(4,3)}$ & [1246.575, 217.713, 216.863, 25.721, 17.733, 3238.238, 11.274, 176.03, 46.546, 3327.307, 61.864, 124.77, -915.328, 6.619, -874.584, $-24.285$, -17.217, 112.124, 104.269, -12.361, -35.605, -48.637, 670.41, 106.28, 643.79, 20.061, 77.684, 53.162, 495.231, 75.412, 471.607, -45.907, -36.852, -48.313, -96.211]$^T$\\ 
  \hline
  $\mathbf{c}_1^{(4,4)}$ & [1.11, 1.288, 1.209, 6.453, -0.895, -322.136, -47.341, 1929.55, 81.366, -304.291, 65.997, 1828.138, -2914.944, -20.701, -2749.887, $-5.079$, -1.989, 29.81, -2.185, -3.866, 8.615, 18.519, -31.436, 14.212, -20.451, -4.566, 0.306, 24.487, -8.508, 5.63, -23.77, -25.062, -18.78, -21.533, -24.54]$^T$\\
  
  \hline 
  $\mathbf{c}_2^{(1,1)}$ & [-8419.001, 492.704, 823.587, 118.032, $-237.465$, 2.095, -43.577, 18.336, 11.661, -20.795, 77.852, 33.126, -45.303, 6.557, -1.614, -0.303, 0.313, -3.416, -0.388, 0.789, 2.731, 0.115, 1.718, -0.413, 0.084, -1.005, -8.073, -1.637, 1.804, -0.523, 0.029, 0.028, 0.097, 0.034, -0.016]$^T$\\ 
  \hline
  $\mathbf{c}_2^{(1,2)}$ & [5228.642, 109.824, 173.029, 50.221, 47.019, 5048.884, 72.53, 14.111, 31.614, 4421.869, 13.191, 6.01, 4455.374, -3.4, 4484.238, 9.44, 8.74, 29.55, 23.576, 6.94, -6.215, -9.259, 254.017, 33.297, 242.619, 9.437, 36.16, 32.45, 322.706, 46.217, 309.732, -37.008, -25.062, $-40.287$, -76.24]$^T$\\ 
  \hline
  $\mathbf{c}_2^{(1,3)}$ & [6989.348, 113.883, 172.083, 54.193, 58.246, 5132.046, 82.207, 35.033, 86.972, 4503.667, 59.403, 27.553, 4577.678, 2.161, 4646.831, 7.465, 7.504, 27.916, 20.56, 5.367, -7.057, -10.89, 253.067, 31.698, 240.594, 8.631, 33.897, 30.403, 320.962, 44.756, 308.47, $-32.755$, -26.104, -41.746, -85.082]$^T$\\ 
  \hline
  $\mathbf{c}_2^{(1,4)}$ & [-22302.131, 4372.126, 1767.365, 23.693, $-49.495$, -234.683, -353.198, -5.043, 6.336, 81.776, -3.51, 5.164, -5.753, 1.307, -3.679, 2.221, 15.658, 0.202, -0.154, 0.967, 0.052, -0.415, 0.18, -0.172, 0.22, -9.471, 0.585, $-0.116$, 0.452, 0.211, 0.14, 0.01, -0.046, 0.008, -0.035]$^T$\\ 
  \hline
  $\mathbf{c}_2^{(2,1)}$ & [-7150.877, -2367.05, 141.575, -1326.002, -1509.024, 412.541, -436.121, -347.226, -384.001, 135.302, -418.265, -400.174, $-1666.273$, -953.099, -1531.119, 38.508, 76.391, 90.031, 94.001, -66.735, -52.465, $-49.50$3, 14.004, 7.192, 12.092, 41.584, 123.24, 116.892, 231.118, 136.039, 216.447, -0.023, -0.001, -0.169, -0.452]$^T$\\ 
  \hline
  $\mathbf{c}_2^{(2,2)}$ & [5246.538, 117.807, 173.256, 51.769, 48.6, 5074.559, 81.403, 18.712, 37.13, 4458.744, 12.48, 6.074, 4494.675, -1.804, 4522.866, 9.106, 8.362, 28.99, 23.124, 6.669, -6.535, -9.539, 255.695, 32.721, 244.287, 9.066, 35.817, 32.192, 321.844, 45.726, 309.015, $-36.536$, -24.876, -40.187, -75.829]$^T$\\ 
  \hline
  $\mathbf{c}_2^{(2,3)}$ & [7006.45, 121.545, 172.366, 55.699, 59.776, 5157.41, 90.899, 39.549, 92.488, 4540.555, 58.647, 27.693, 4617.027, 3.82, 4685.4, 7.129, 7.131, 27.37, 20.125, 5.085, -7.362, -11.171, 254.756, 31.137, 242.269, 8.266, 33.585, 30.171, 320.113, 44.28, 307.777, -32.277, $-25.912$, -41.638, -84.669]$^T$\\ 
  \hline
  $\mathbf{c}_2^{(2,4)}$ & [-7350.338, -4863.528, -278.533, -617.352, -584.947, 403.027, -296.435, -401.974, -457.226, 166.258, -358.676, -354.013, $-462.807$, -384.38, -381.004, 99.425, 75.303, 113.771, 127.419, -78.763, -89.81, -75.296, 17.216, 12.549, 7.615, 61.094, 112.125, 101.312, 56.658, 49.062, 52.064, -0.105, 0.354, 0.305, -1.127]$^T$\\ 
  \hline
  $\mathbf{c}_2^{(3,1)}$ & [-13045.134, 284.266, 225.737, -2771.228, -2945.271, -43.614, -175.508, -95.136, $-96.347$, 8.845, -111.175, -107.817, $-3566.451$, -1110.341, -3467.507, 58.475, 8.575, 73.807, 76.306, 4.557, -23.573, -25.989, 160.108, 49.195, 154.886, 17.942, 66.915, 65.889, 172.589, 65.779, 164.742, -43.441, -19.431, -18.373, -42.248]$^T$\\ 
  \hline
  $\mathbf{c}_2^{(3,2)}$ & [2974.976, 256.59, 230.904, 6.364, 0.826, 4920.3, 39.844, 10.333, 9.42, 4228.165, 11.818, 10.586, 2286.43, -51.192, 2307.688, 14.518, 11.756, 35.542, 29.484, 11.624, -11.999, -13.999, 348.139, 49.251, 334.839, 12.825, 37.471, 33.05, 346.945, 56.911, 331.967, -42.63, -29.304, -41.656, -81.055]$^T$\\ 
  \hline
  $\mathbf{c}_2^{(3,3)}$ & [4814.16, 256.806, 231.143, 13.273, 14.828, 5006.129, 50.938, 18.117, 18.093, 4316.45, 20.497, 18.858, 2411.738, -45.045, 2473.004, 12.249, 10.044, 34.048, 27.143, 9.554, -12.237, -14.191, 347.439, 49.159, 333.789, 11.643, 36.417, 30.841, 346.239, 56.783, 330.912, $-38.481$, -30.419, -43.185, -89.956]$^T$\\ 
  \hline
  $\mathbf{c}_2^{(3,4)}$ & [-23980.403, 1361.739, 1677.576, -1650.337, -1610.656, -353.646, -277.568, -342.467, $-328.666$, -617.679, -491.073, -471.172, $-855.141$, -506.421, -801.869, 157.477, -94.885, 101.743, 97.974, 44.797, -62.051, -58.772, 61.579, 33.215, 58.286, 127.944, 137.606, 130.932, 85.54, 52.389, 80.619, -1.164, 1.233, 1.232, -0.659]$^T$\\ 
  \hline
  $\mathbf{c}_2^{(4,1)}$ & [-14908.929, 99.126, 95.786, -3419.89, $-3533.21$, 52.217, -17.979, -71.209, -67.517, 45.68, -71.002, -66.591, -5193.195, -1201.831, -5108.697, -0.068, 1.974, 32.357, 33.303, -4.271, -32.805, -33.09, 181.616, 52.396, 175.496, 0.517, 29.157, 29.615, 184.148, 59.592, 177.824, -49.465, -17.3, -16.058, $-48.095$]$^T$\\ 
  \hline
  $\mathbf{c}_2^{(4,2)}$ & [2064.193, 258.159, 232.587, 51.562, 39.436, 4979.877, 30.208, 10.619, 8.82, 4292.236, 12.598, 10.525, 1389.9, 26.852, 1393.538, 14.587, 9.996, 41.944, 35.057, 10.071, -13.995, -16.162, 385.625, 53.365, 370.999, 12.632, 42.2, 37.14, 385.685, 60.453, 369.653, -45.604, -31.793, -43.5, -83.263]$^T$\\ 
  \hline
  $\mathbf{c}_2^{(4,3)}$ & [3670.301, 258.716, 233.121, 59.014, 51.885, 5028.702, 39.383, 14.466, 12.301, 4348.359, 16.787, 13.941, 1504.763, 33.617, 1543.853, 13.404, 9.178, 41.496, 33.649, 8.774, -14.188, -16.586, 385.269, 53.413, 370.356, 12.162, 41.727, 35.77, 385.241, 60.44, 368.959, $-41.179$, -32.656, -44.704, -91.848]$^T$\\ 
  \hline
  $\mathbf{c}_2^{(4,4)}$ & [2.241, 0.6, 0.556, 5.484, 1.583, 0.135, $-0.78$, 0.406, 0.668, 0.169, 0.03, 0.196, -20.263, 4.658, -11.003, 0.047, 0.05, -0.206, 0.048, 0.007, -0.482, -0.491, 3.839, 1.204, 3.326, 0.014, -0.082, -0.065, 4.247, 0.902, 2.685, $-14.239$, -5.733, -5.799, -13.154]$^T$\\
  \hline
\end{supertabular}

\end{document}